\documentclass[10pt]{article}

\usepackage[T1]{fontenc}
\usepackage{pst-node}

\usepackage{amsmath}
\usepackage{amsthm}
\usepackage{amssymb}
\numberwithin{equation}{section}

\theoremstyle{plain}
\newtheorem{proposition}{Proposition}[section]
\newtheorem{corollary}[proposition]{Corollary}
\newtheorem{lemma}[proposition]{Lemma}
\newtheorem{theorem}[proposition]{Theorem}

\theoremstyle{definition}
\newtheorem{definition}[proposition]{Definition}
\newtheorem{example}[proposition]{Example}
\newtheorem{remark}[proposition]{Remark}

\newcommand{\ds}{\displaystyle}
\newcommand{\R}{\mathbf{R}}
\newcommand{\smallfrac}[2]{{\textstyle\frac{#1}{#2}}}

\newcommand{\abs}[1]{\left\lvert #1 \right\rvert}
\newcommand{\avg}[1]{\bigl\langle #1 \bigr\rangle}
\newcommand{\jump}[1]{\bigl[ #1 \bigr]}
\newcommand{\poisson}[2]{\left\{ #1,#2 \right\}}
\newcommand{\detD}[2]{\mathcal{D}^{(#1)}_{#2}}
\newcommand{\detDprime}[1]{\mathcal{D}'_{#1}}
\newcommand{\puresector}{\mathcal{P}}
\newcommand{\piecewiseclass}{PC^{\infty}}
\newcommand{\spaceDprime}{\mathcal{D}'(\R)}

\DeclareMathOperator{\sgn}{sgn}

\DeclareMathOperator{\diag}{diag}

\begin{document}

\title{Explicit multipeakon solutions of Novikov's cubically nonlinear integrable Camassa--Holm type equation}
\author{%
  Andrew N. W. Hone\thanks{Institute of Mathematics, Statistics \& Actuarial Science, University of Kent, Canterbury CT2 7NF, United Kingdom; anwh@kent.ac.uk}
  \and
  Hans Lundmark\thanks{Department of Mathematics, Link{\"o}ping University, SE-581 83 Link{\"o}ping, Sweden; halun@mai.liu.se}
  \and
  Jacek Szmigielski\thanks{Department of Mathematics and Statistics, University of Saskatchewan, 106 Wiggins Road, Saskatoon, Saskatchewan, S7N 5E6, Canada; szmigiel@math.usask.ca}
}

\date{March 20, 2009}
%\date{\today}

\maketitle

\begin{abstract}
  Recently Vladimir Novikov found a new integrable analogue of the
  Camassa--Holm equation, admitting peaked soliton (\emph{peakon})
  solutions, which has nonlinear terms that are cubic, rather than
  quadratic. In this paper, the explicit formulas for multipeakon
  solutions of Novikov's cubically nonlinear equation are calculated,
  using the matrix Lax pair found by Hone and Wang. By a
  transformation of Liouville type, the associated spectral problem is
  related to a cubic string equation, which is dual to the cubic
  string that was previously found in the work of Lundmark and
  Szmigielski on the multipeakons of the Degasperis--Procesi equation.
\end{abstract}

\section{Introduction}

Integrable PDEs with nonsmooth solutions have attracted much attention
in recent years, since the discovery of the Camassa--Holm shallow water
wave equation and its peak-shaped soliton solutions called \emph{peakons}
\cite{camassa-holm}.
Our purpose in this paper is to explicitly compute
the multipeakon solutions of a new integrable PDE,
equation \eqref{eq:novikov} below,
which is of the Camassa--Holm form $u_t - u_{xxt} = F(u,u_x,u_{xx},\dots)$,
but has cubically nonlinear terms instead of quadratic.
This equation was found by Vladimir Novikov, and published in a recent
paper by Hone and Wang \cite{hone-wang-cubic-nonlinearity}.

We will apply inverse spectral methods. The spatial equation in the
Lax pair for Novikov's equation turns out to be equivalent to what we
call the \emph{dual cubic string}, a spectral problem closely related
to the \emph{cubic string} that was used for finding the multipeakon
solutions to the Degasperis--Procesi equation
\cite{lundmark-szmigielski-DPshort,lundmark-szmigielski-DPlong,kohlenberg-lundmark-szmigielski}.
Once this relation is established, the Novikov peakon solution can be
derived in a straightforward way using the results obtained in
\cite{kohlenberg-lundmark-szmigielski}.
The constants of motion have a more complicated structure than in the
Camassa--Holm and Degasperis--Procesi cases, and the study
of this gives as an interesting by-product
a combinatorial identity concerning the sum of all minors in a symmetric matrix,
which we have dubbed the \emph{Canada Day Theorem}
(Theorem~\ref{thm:CanadaDay}, proved in Appendix~\ref{sec:combinatorics}).

The peakon problem for Novikov's equation presents in addition one
important challenge. Unlike its Camassa--Holm or Degasperis--Procesi
counterparts, the Lax pair for the Novikov equation is originally
ill-defined in the peakon sector. The problem is caused by terms which
involve multiplication of a singular measure by a discontinuous
function. We prove in Appendix~\ref{sec:multiplication} that there
exists a regularization of the Lax pair which preserves integrability
of the peakon sector, thus allowing us to use spectral and inverse
spectral methods to obtain the multipeakon solutions to the Novikov
equation. This regularization problem has a subtle but
nevertheless real impact on the formulas.
In general, the use of Lax pairs to construct distributional solutions
to nonlinear equations which are Lax integrable in the smooth sector
but may not be so in the whole non-smooth sector is relatively
uncharted territory, and the case of Novikov's equation may provide
some relevant insight in this regard.

\section{Background}

The main example of a PDE admitting peaked solitons is the family
\begin{equation}
  \label{eq:b-family}
  u_t - u_{xxt} + (b+1) u u_x = b u_x u_{xx} + u u_{xxx},
\end{equation}
often written as
\begin{equation}
  \label{eq:b-family-m}
  m_t + m_x u + b m u_x = 0, \qquad m = u - u_{xx},
\end{equation}
which was introduced by Degasperis, Holm and Hone \cite{degasperis-holm-hone},
and is Hamiltonian for all values of $b$ \cite{holm-hone}.
It includes the Camassa--Holm equation as the case $b=2$,
and another integrable PDE called the
Degasperis--Procesi equation \cite{degasperis-procesi,degasperis-holm-hone}
as the case $b=3$.
These are the only values of $b$ for which the equation is integrable,
according to a variety of integrability tests
\cite{degasperis-procesi,mikhailov-novikov-perturbative,hone-wang-prolongation-algebras,ivanov-integrability-test}.
(However, we note that the case $b=0$ is excluded from the aforementioned
integrability tests; yet this case provides a regularization
of the inviscid Burgers equation that is Hamiltonian and has classical
solutions globally in time \cite{bhat-bzero}.)
\emph{Multipeakons} are weak solutions of the form
\begin{equation}
  \label{eq:peakon-ansatz}
  u(x,t) = \sum_{i=1}^n m_i(t) \, e^{-\abs{x-x_i(t)}},
\end{equation}
formed through superposition of $n$ \emph{peakons}
(peaked solitons of the shape $e^{-\abs{x}}$).
This ansatz satisfies the PDE \eqref{eq:b-family-m} if and only if the
positions $(x_1,\dots,x_n)$ and momenta $(m_1,\dots,m_n)$ of the
peakons obey the following system of $2n$ ODEs:
\begin{equation}
  \label{eq:b-peakon-ode}
  \dot x_k = \sum_{i=1}^n m_i \, e^{-\abs{x_k-x_i}},
  \qquad
  \dot m_k = (b-1) \, m_k \sum_{i=1}^n m_i \, \sgn(x_k-x_i) \, e^{-\abs{x_k-x_i}}.
\end{equation}
Here, $\sgn x$ denotes the signum function, which is $+1$, $-1$ or $0$
depending on whether $x$ is positive, negative or zero.
In shorthand notation,
with $\avg{f(x)}$ denoting the average of the left and right limits,
\begin{equation}
  \label{eq:average-notation}
  \avg{f(x)} = \frac12 \bigl( f(x^-)+f(x^+) \bigr),
\end{equation}
the ODEs can be written as
\begin{equation}
  \label{eq:b-ode-short}
  \dot x_k = u(x_k),
  \qquad
  \dot m_k = -(b-1) \, m_k \, \avg{u_x(x_k)}.
\end{equation}
In the Camassa--Holm case $b=2$, this is a canonical Hamiltonian system
generated by $h=\frac12 \sum_{j,k=1}^n m_j \, m_k \, e^{-\abs{x_j-x_k}}$.
Explicit formulas for the $n$-peakon solution of the Camassa--Holm equation
were derived by Beals, Sattinger and Szmigielski \cite{beals-sattinger-szmigielski-stieltjes,beals-sattinger-szmigielski-moment}
using inverse spectral methods,
and the same thing for the Degasperis--Procesi equation
was accomplished by Lundmark and Szmigielski
\cite{lundmark-szmigielski-DPshort,lundmark-szmigielski-DPlong}.

It requires some care to specify the exact sense in which the peakon
solutions satisfy the PDE.
The formulation \eqref{eq:b-family-m} suffers from the problem that
the product $mu_x$ is ill-defined in the peakon case, since the quantity
$m=u-u_{xx}=2 \sum_{i=1}^n m_i \, \delta_{x_i}$ is a discrete measure,
and it is multiplied by a function $u_x$ which has jump discontinuities
exactly at the points $x_k$ where the Dirac deltas in the measure $m$
are situated.
To avoid this problem, one can instead rewrite \eqref{eq:b-family} as
\begin{equation}
  \label{eq:b-family-alt}
  (1-\partial_x^2) u_t
  + (b+1-\partial_x^2) \, \partial_x \left( \smallfrac12 \, u^2 \right)
  + \partial_x \left( \smallfrac{3-b}{2} \, u_x^2 \right)
  = 0.
\end{equation}
Then a function $u(x,t)$ is said to be a solution if
\begin{itemize}
\item
  $u(\cdot,t) \in W^{1,2}_{\mathrm{loc}}(\R)$ for each fixed~$t$,
  which means that $u(\cdot,t)^2$ and $u_x(\cdot,t)^2$ are locally integrable
  functions, and therefore define distributions of class $\spaceDprime$
  (i.e., continuous linear functionals acting on compactly supported
  $C^{\infty}$ test functions on the real line~$\R$),
\item
  the time derivative $u_t(\cdot,t)$, defined as the limit
  of a difference quotient, exists as a distribution in
  $\spaceDprime$
  for all~$t$,
\item
  equation \eqref{eq:b-family-alt},
  with $\partial_x$ taken to mean the usual distributional derivative,
  is satisfied for all~$t$ in the sense of distributions in
  $\spaceDprime$.
\end{itemize}
It is worth mentioning that functions in the space
$W^{1,2}_{\mathrm{loc}}(\R)$
are continuous, by the Sobolev embedding theorem.
However, the term $u_x^2$ is absent from equation \eqref{eq:b-family-alt}
if $b=3$, so in that particular case one requires only that
$u(\cdot,t) \in L^2_{\mathrm{loc}}(\R)$; this means that the
Degasperis--Procesi can admit solutions $u$ that are not continuous
\cite{coclite-karlsen-DPwellposedness,coclite-karlsen-DPuniqueness,lundmark-shockpeakons}.

It is often appropriate to rewrite equation \eqref{eq:b-family-alt} as
a nonlocal evolution equation for~$u$ by inverting the operator
$(1-\partial_x^2)$, as was done in \cite{constantin-escher,constantin-mckean}
for the Camassa--Holm equation. However, the distributional
formulation used here is very convenient when working with peakon
solution.

\section{Novikov's equation}
\label{sec:novikovs-equation}

The new integrable equation found by Vladimir Novikov is
\begin{equation}
  \label{eq:novikov}
  u_t - u_{xxt} + 4 u^2 u_x = 3 u u_x u_{xx} + u^2 u_{xxx},
\end{equation}
which can be written as
\begin{equation}
  \label{eq:novikov-system}
  m_t + (m_x u + 3 m u_x) \, u = 0, \qquad m = u - u_{xx},
\end{equation}
to highlight the similarity in form to the Degasperis--Procesi equation,
or as
\begin{equation}
  \label{eq:novikov-alt}
  (1-\partial_x^2) u_t
  + (4-\partial_x^2) \, \partial_x \left( \smallfrac13 \, u^3 \right)
  + \partial_x \left( \smallfrac{3}{2} \, u u_x^2 \right)
  + \smallfrac12 \, u_x^3
  = 0
\end{equation}
in order to rigorously define weak solutions as above,
except that here one requires that
$u(\cdot,t) \in W^{1,3}_{\mathrm{loc}}(\R)$ for all~$t$,
so that $u^3$ and $u_x^3$ are locally integrable and therefore define
distributions in $\spaceDprime$;
it then follows from H\"older's inequality with the conjugate indices
$3$ and $3/2$ that $u u_x^2$ is locally integrable as well,
and \eqref{eq:novikov-alt} can thus be interpreted as a distributional
equation.
Since functions in $W^{1,p}_{\mathrm{loc}}(\R)$ with $p\geq 1$ are
automatically continuous, Novikov's equation
is similar to the Camassa--Holm equation in that it
only admits continuous distributional solutions
(as opposed to the Degasperis--Procesi equation,
which has discontinuous solutions as well).

Like the equations in the $b$-family \eqref{eq:b-family},
Novikov's equation admits (in the weak sense just defined)
multipeakon solutions of the form \eqref{eq:peakon-ansatz},
but in this case the ODEs for the positions and momenta are
\begin{equation}
  \label{eq:novikov-ode}
  \begin{split}
    \dot x_k &= u(x_k)^2
    = \left( \sum_{i=1}^n m_i \, e^{-\abs{x_k-x_i}} \right)^2,
    \\
    \dot m_k &= - m_k \, u(x_k) \, \avg{u_x(x_k)} \\
    &= m_k \, \left( \sum_{i=1}^n m_i \, e^{-\abs{x_k-x_i}} \right)
    \left( \sum_{j=1}^n m_j \, \sgn(x_k-x_j) \, e^{-\abs{x_k-x_j}} \right).
  \end{split}
\end{equation}
These equations were stated in \cite{hone-wang-cubic-nonlinearity},
where it was also shown that they constitute a Hamiltonian system
$\dot x_k = \poisson{x_k}{h}$,
$\dot m_k = \poisson{m_k}{h}$,
generated by the same Hamiltonian
$h = \frac12 \sum_{j,k=1}^n m_j m_k \, e^{-\abs{x_j-x_k}}$
as the Camassa--Holm peakons,
but with respect to a different, non-canonical, Poisson structure given by
\begin{equation}
  \label{eq:poisson-structure}
  \begin{split}
    \poisson{x_j}{x_k} &= \sgn(x_j-x_k) \, \bigl( 1 - E_{jk}^2 \bigr),\\
    \poisson{x_j}{m_k} &= m_k E_{jk}^2,\\
    \poisson{m_j}{m_k} &= \sgn(x_j-x_k) \, m_j m_k E_{jk}^2,
    \qquad\text{where $E_{jk} = e^{-\abs{x_j-x_k}}$.}
  \end{split}
\end{equation}
As will be shown below,
\eqref{eq:novikov-ode} is a Liouville integrable system
(Theorem~\ref{thm:constants-of-motion-commute});
in fact, it is even explicitly solvable in terms of elementary functions
(Theorem~\ref{thm:explicit_peakons}).

\section{Forward spectral problem}
\label{sec:forward}

In order to integrate the Novikov peakon ODEs,
we are going to make use of the matrix Lax pair
found by Hone and Wang \cite{hone-wang-cubic-nonlinearity},
specified by the following matrix linear system:
\begin{equation}
  \label{eq:lax-x}
    \frac{\partial}{\partial x}
    \begin{pmatrix} \psi_1 \\ \psi_2 \\ \psi_3 \end{pmatrix} =
    \begin{pmatrix}
      0 & zm & 1 \\
      0 & 0 & zm \\
      1 & 0 & 0
    \end{pmatrix}
    \begin{pmatrix} \psi_1 \\ \psi_2 \\ \psi_3 \end{pmatrix},
\end{equation}
\begin{equation}
  \label{eq:lax-t}
    \frac{\partial}{\partial t}
    \begin{pmatrix} \psi_1 \\ \psi_2 \\ \psi_3 \end{pmatrix} =
    \begin{pmatrix}
      -uu_x & u_x z^{-1}-u^2mz & u_x^2 \\
      u z^{-1} & -z^{-2} & -u_x z^{-1}-u^2mz \\
      -u^2 & u z^{-1} & uu_x
    \end{pmatrix}
    \begin{pmatrix} \psi_1 \\ \psi_2 \\ \psi_3 \end{pmatrix}.
\end{equation}
(Compared with reference \cite{hone-wang-cubic-nonlinearity}
we have added a constant
multiple of the identity to the matrix on the right hand side
of \eqref{eq:lax-t}, and used $z$ in place of $\lambda$.)
In the peakon case, when $u = \sum_{i=1}^n m_i \, e^{-\abs{x-x_i}}$,
the quantity $m=u-u_{xx}=2 \sum_{i=1}^n m_i \, \delta_{x_i}$
is a discrete measure.
We assume that $x_1 < x_2 < \dots < x_n$ (which at least remains true
for a while if it is true for $t=0$). These points divide the $x$ axis into
$n+1$ intervals which we number from $0$ to~$n$, so that the $k$th interval
runs from $x_k$ to $x_{k+1}$, with the convention that
$x_0=-\infty$ and $x_{n+1}=+\infty$.
Since $m$ vanishes between the point masses,
equation \eqref{eq:lax-x} reduces to
$\partial_x \psi_1=\psi_3$, $\partial_x \psi_2=0$ and $\partial_x \psi_3=\psi_1$
in each interval, so that in the $k$th interval we have
\begin{equation}
  \label{piecewise-x}
  \begin{pmatrix} \psi_1 \\ \psi_2 \\ \psi_3 \end{pmatrix} =
  \begin{pmatrix}
    A_k \, e^x + z^2 \, C_k \, e^{-x} \\ 2z \, B_k \\ A_k \, e^x - z^2 \, C_k \, e^{-x}
  \end{pmatrix}
  \qquad\text{for $x_k < x < x_{k+1}$},
\end{equation}
where the factors containing $z$ have been inserted for later convenience.
These piecewise solutions are then glued together at the points $x_k$.
The proper interpretation of \eqref{eq:lax-x} at these points turns
out to be
that $\psi_3$ must be continuous, while $\psi_1$ and
$\psi_2$ are allowed to have jump discontinuities; moreover, in the term
$zm\psi_2$, one should take $\psi_2(x) \delta_{x_k}$ to mean
$\avg{\psi_2(x_k)} \delta_{x_k}$. This point is fully explained in Appendix~\ref{sec:multiplication}.  This leads to
\begin{equation}
  \label{eq:jump-matrix}
  \begin{split}
    \begin{pmatrix} A_k \\ B_k \\ C_k \end{pmatrix} &=
    \begin{pmatrix}
      1 - \lambda \, m_k^2 & -2\lambda \, m_k \, e^{-x_k}  & -\lambda^2 \, m_k^2 \, e^{-2x_k} \\
      m_k \, e^{x_k} & 1 & \lambda \, m_k \, e^{-x_k} \\
      m_k^2 \, e^{2x_k} & 2 \, m_k \, e^{x_k} & 1 + \lambda m_k^2
    \end{pmatrix}
    \begin{pmatrix} A_{k-1} \\ B_{k-1} \\ C_{k-1} \end{pmatrix}
    \\
    &=: S_k(\lambda)
    \begin{pmatrix} A_{k-1} \\ B_{k-1} \\ C_{k-1} \end{pmatrix},
    \qquad\text{where $\lambda=-z^2$}.
  \end{split}
\end{equation}
We impose the boundary condition $(A_0,B_0,C_0)=(1,0,0)$,
which is consistent with the
time evolution given by \eqref{eq:lax-t} for $x<x_1$.
Then all $(A_k,B_k,C_k)$ are determined by successive application
of the jump matrices $S_k(\lambda)$ as in \eqref{eq:jump-matrix}.
For $x>x_n$, equation \eqref{eq:lax-t} implies that
$(A,B,C):=(A_n,B_n,C_n)$
evolves as
\begin{equation}
  \label{eq:ABCevolution}
  \dot A = 0,\quad
  \dot B = \frac{B-A M_+}{\lambda},\quad
  \dot C = \frac{2 M_+ \, (B - A M_+)}{\lambda},
\end{equation}
where $M_+ = \sum_{k=1}^N m_k \, e^{x_k}$.
Thus $A$ is invariant. It is the $(1,1)$ entry of the total jump matrix
\begin{equation}
  \label{eq:S}
  S(\lambda)=S_n(\lambda) \dots S_2(\lambda) S_1(\lambda),
\end{equation}
and therefore it is a polynomial in $\lambda$ of degree~$n$,
\begin{equation}
  \label{eq:A}
  A(\lambda) = \sum_{k=0}^n H_k (-\lambda)^k
  =\left( 1 - \frac{\lambda}{\lambda_1} \right) \dots \left( 1 - \frac{\lambda}{\lambda_n} \right),
\end{equation}
where $H_0=1$ (since $S(0)=I$, the identity matrix),
and where the other coefficients $H_1,\dots,H_n$ are
Poisson commuting constants of motion
(see Theorems~\ref{thm:constants-of-motion}
and~\ref{thm:constants-of-motion-commute} below).

The first linear equation \eqref{eq:lax-x}, together with
the boundary conditions expressed by the requirements that
$B_0=C_0=0$ and $A_n(\lambda)=0$,
is a spectral problem which has the zeros
$\lambda_1,\dots,\lambda_n$ of $A(\lambda)$
as its eigen\-values.
(To be precise, one should perhaps say that it is the corresponding
values of $z=\pm\sqrt{-\lambda}$ that are the eigenvalues,
but we will soon show that the $\lambda_k$ are positive,
at least in the pure peakon case,
and therefore more convenient to work with than the purely imaginary
values of~$z$; see \eqref{eq:ordering-lambda} below.)

Elimination of $\psi_1$ from \eqref{eq:lax-x} gives
$\partial_x \psi_2 = zm \psi_3$ and
$(\partial_x^2-1)\,\psi_3 = zm\psi_2$,
and the boundary conditions above imply that
$(\psi_2,\psi_3) \to (0,0)$ as $x\to-\infty$
and $\psi_3 \to 0$ as $x\to+\infty$.
Using the Green's function $-e^{-\abs{x}}/2$ for the operator $\partial_x^2-1$
with vanishing boundary conditions, we can rephrase the problem as a system
of integral equations,
\begin{equation}
  \label{eq:integral-eqns}
  \begin{split}
    \psi_2(x) &= z \int_{-\infty}^x \psi_3(y) \, dm(y),\\
    \psi_3(x) &= -z \int_{-\infty}^{\infty} \frac12 \, e^{-\abs{x-y}} \psi_2(y) \, dm(y),\\
  \end{split}
\end{equation}
with integrals taken with respect to the discrete measure
$m=2\sum_{i=1}^n m_i \, \delta_{x_i}$.
Here, there is again the problem of Dirac deltas multiplying a function
$\psi_2$ with jump discontinuities, and we take
$\psi_2(x) \delta_{x_k}$ to mean the average $\avg{\psi_2(x_k)} \delta_{x_k}$,
in full agreement with the earlier definition of the singular term
appearing in the spectral problem.
Then
\begin{equation}
  \label{eq:integral-eqns-avg}
  \begin{split}
    \avg{\psi_2(x_j)} &= z \left( 2 \sum_{k=1}^{j-1} \psi_3(x_k)\,m_k + \psi_3(x_j)\,m_j \right),\\
    \psi_3(x_j) &= -z \sum_{k=1}^n e^{-\abs{x_j-x_k}} \avg{\psi_2(x_k)} \, m_k,\\
  \end{split}
\end{equation}
which can be written in block matrix notation as
\begin{equation}
  \label{eq:block-matrix}
  \begin{pmatrix} \avg{\Psi_2} \\ \Psi_3 \end{pmatrix}
  = z \begin{pmatrix} 0 & TP \\ -EP & 0 \end{pmatrix}
  \begin{pmatrix} \avg{\Psi_2} \\ \Psi_3 \end{pmatrix},
\end{equation}
where
\begin{equation}
  \label{eq:matrix-notation}
  \begin{split}
    \Psi_3 &= \bigl(\psi_3(x_1),\dots,\psi_3(x_n) \bigr)^t,\\
    \avg{\Psi_2} &= \bigl(\avg{\psi_2(x_1)},\dots,\avg{\psi_2(x_n)} \bigr)^t,\\
    P &= \diag(m_1,\dots,m_n),\\
    E &= (E_{jk})_{j,k=1}^n, \qquad\text{where $E_{jk}=e^{-\abs{x_j-x_k}}$},\\
    T &= (T_{jk})_{j,k=1}^n, \qquad\text{where $T_{jk}=1+\sgn(j-k)$}.\\
  \end{split}
\end{equation}
(In words, $T$ is the lower triangular $n\times n$ matrix that has $1$ on
the main diagonal and $2$ everywhere below it.)
In terms of $\avg{\Psi_2}$ alone, we have
\begin{equation}
  \label{eq:TPEP-appears-here}
  \avg{\Psi_2} = -z^2 TPEP \avg{\Psi_2},
\end{equation}
so the eigenvalues are given by
$0=\det(I + z^2 TPEP)=\det(I - \lambda TPEP)$,
where of course $I$ denotes the $n\times n$ identity matrix.
Since the eigenvalues are the zeros of $A(\lambda)$,
and since $A(0)=1$, it follows that
\begin{equation}
  \label{eq:Adet}
  A(\lambda) = \det(I - \lambda TPEP).
\end{equation}
This gives us a fairly concrete representation of the constants of motion
$H_k$, which by definition are the coefficients of~$A(\lambda)$
(see \eqref{eq:A}),
and it can be made even more explicit thanks to the curious combinatorial
result in Theorem~\ref{thm:CanadaDay}.
We remind the reader that a $k\times k$ \emph{minor} of an
$n\times n$ matrix $X$ is,
by definition, the
determinant of a submatrix $X_{IJ} = (X_{ij})_{i \in I, \, J \in J}$
whose rows and columns are selected among those of $X$ by two index
sets $I, J \subseteq \{ 1,\dots,n \}$ with $k$ elements each, and a
\emph{principal minor} is one for which $I=J$.
Compare the result of the theorem with the well-known fact that the
coefficient of $s^k$ in $\det(I+sX)$ equals the sum of all
\emph{principal} $k\times k$ minors of~$X$, regardless of whether $X$
is symmetric or not.

\begin{theorem}[``The Canada Day Theorem'']
  \label{thm:CanadaDay}
  Let the matrix $T$ be defined as in \eqref{eq:matrix-notation} above.
  Then, for any symmetric $n\times n$ matrix~$X$,
  the coefficient of $s^k$ in the polynomial $\det(I+s\,TX)$
  equals the sum of all $k\times k$ minors (principal and non-principal)
  of~$X$.
\end{theorem}

\begin{proof}
  The proof is presented in Appendix~\ref{sec:combinatorics}.
  It relies on the Cauchy--Binet formula, Lindstr\"om's Lemma,
  and some rather intricate dependencies among the minors of~$X$
  due to the symmetry of the matrix.
\end{proof}

Theorem~\ref{thm:CanadaDay} is named after the date
when we started trying to prove it: July~1, 2008, Canada's national day.
(It turned out that the proof was more difficult than we expected,
so we didn't finish it until a few days later.)
Summarizing the results so far,
we now have the following description of the constants of motion:

\begin{theorem}
  \label{thm:constants-of-motion}
  The Novikov peakon ODEs \eqref{eq:novikov-ode} admit $n$
  constants of motion $H_1,\dots,H_n$,
  where $H_k$ equals the sum of all $k\times k$ minors (principal
  and non-principal) of the $n\times n$ symmetric matrix
  $PEP=(m_j m_k E_{jk})_{j,k=1}^n$.
  (See \eqref{eq:matrix-notation} for notation.)
\end{theorem}

\begin{proof}
  This follows at once from \eqref{eq:A}, \eqref{eq:Adet},
  and Theorem~\ref{thm:CanadaDay}.
\end{proof}

\begin{example}
  The sum of all $1\times 1$ minors of $PEP$ is
  of course just the sum of all entries,
  \begin{equation}
    \label{eq:H1}
    H_1 = \sum_{j,k=1}^n m_j m_k E_{jk}
    = \sum_{j,k=1}^n m_j m_k \, e^{-\abs{x_j-x_k}},
  \end{equation}
  and the Hamiltonian of the peakon ODEs \eqref{eq:novikov-ode} is
  $h = \frac12 H_1$.
  Higher order minors of $PEP$ are easily computed using Lindstr\"om's Lemma,
  as explained in Section~\ref{sec:PEP-minors} in the appendix.
  In particular, the constant of motion of highest degree in the $m_k$ is
  \begin{equation}
    \label{eq:Hn}
    H_n = \det(PEP) = \prod_{j=1}^{n-1} (1-E_{j,j+1}^2) \, \prod_{j=1}^{n} m_j^2.
  \end{equation}
\end{example}

\begin{example}
  Written out in full, the constants of motion in the case $n=3$ are
  \begin{equation}
    \label{eq:Hk-n3}
    \begin{split}
      H_1 &= m_1^2 + m_2^2 + m_3^2 + 2 m_1 m_2 E_{12} + 2 m_1 m_3 E_{13} + 2 m_2 m_3 E_{23}, \\
      H_2 &= (1-E_{12}^2) \, m_1^2 \, m_2^2 + (1-E_{13}^2) \, m_1^2 \, m_3^2 + (1-E_{23}^2) \, m_2^2 \, m_3^2 \\
      & \quad + 2 (E_{23}-E_{12}E_{13}) \, m_1^2 \, m_2 \, m_3 + 2 (E_{12}-E_{13}E_{23}) \, m_1 \, m_2 \, m_3^2, \\
      H_3 &= (1-E_{12}^2)(1-E_{23}^2) \, m_1^2 \, m_2^2 \, m_3^2.
    \end{split}
  \end{equation}
\end{example}

From now on we mainly restrict ourselves to the pure peakon case when
$m_k > 0$ for all~$k$ (no antipeakons). Our first reason for this is that we
can then use the positivity of $H_1$ and~$H_n$ to show global
existence of peakon solutions.

\begin{theorem}
  \label{thm:globalexistence}
  Let
  \begin{equation}
    \label{eq:phasespaceM}
    \puresector= \{ x_1 < \dots < x_n, \,\, \text{all $m_k > 0$} \}
  \end{equation}
  be the phase space for the Novikov peakon system \eqref{eq:novikov-ode} in
  the pure peakon case. If the initial data are in $\puresector$, then
  the solution $(\mathbf{x}(t), \mathbf{m}(t))$ exists
  for all $t \in \R$,
  and remains in~$\puresector$.
\end{theorem}

\begin{proof}
  Local existence in $\puresector$ is automatic in view of the
  smoothness of the ODEs there. By \eqref{eq:H1} and \eqref{eq:Hn},
  both $H_1$ and $H_n$ are strictly positive on~$\puresector$. Since
  $m_k^2 < H_1$, all $m_k$ remain bounded from above. The positivity
  of $H_n$ ensures that the $m_k$ are bounded away from zero, and
  that the positions remain ordered.
  The velocities $\dot x_k$ are all bounded by $(\sum m_k)^2$,
  hence $0 < \dot x_k \le C$ for some constant~$C$,
  and the positions $x_k(t)$ are therefore finite
  for all $t \in \R$.
  Since neither $x_k$ nor $m_k$ can blow up in finite time,
  the solution exists globally in time.
\end{proof}

\begin{remark}
  The peakon ODEs \eqref{eq:novikov-ode} are invariant under
  the transformation
  $(m_1,\dots,m_n) \mapsto (-m_1,\dots,-m_n)$,
  so the analogous result holds also when all $m_k$ are negative.
\end{remark}

\begin{theorem}
  \label{thm:constants-of-motion-commute}
  The constants of motion $H_1,\dots,H_n$ of
  Theorem~\ref{thm:constants-of-motion} are functionally independent
  and commute with respect to the Poisson bracket
  \eqref{eq:poisson-structure}, so the Novikov peakon system
  \eqref{eq:novikov-ode} is Liouville integrable on the phase
  space~$\puresector$.
\end{theorem}

\begin{proof}
  To prove functional independence, one should check that
  $J:=dH_1\wedge dH_2\wedge \ldots \wedge dH_n$
  does not vanish on any open set in~$\puresector$.
  Since $J$ is rational in the variables $\{ m_k, e^{x_k} \}_{k=1}^n$,
  it vanishes identically if it vanishes on an open set,
  so it is sufficient to show that $J$ is not identically zero.
  To see this, note that
  \begin{equation}
    \label{eq:hksym}
    H_k=e_k(m_1^2,\ldots ,m_n^2)+O(E_{pq}),
  \end{equation}
  where $e_k$ denotes the $k$th elementary symmetric function
  in $n$ variables, and $O(E_{pq})$ denotes
  terms involving exponentials of the positions $x_j$. It is well
  known that the first $n$ elementary symmetric functions are
  independent (they provide a basis for symmetric functions of $n$
  variables \cite{macdonald}), and therefore the leading part of $J$
  (neglecting the $O(E_{pq})$ terms) does not vanish.
  Since the $O(E_{pq})$ terms can be made arbitrarily small
  by taking the $x_k$ far apart, we see that there is a region in
  $\puresector$ where $J$ does not vanish, and we are done.

  To prove that the quantities $H_k$ Poisson commute with respect to
  the bracket \eqref{eq:poisson-structure}, it is convenient to adapt
  some arguments of Moser that he applied to the scattering of particles
  in the Toda lattice and the rational Calogero--Moser system
  \cite{moser-three}. The Poisson bracket of two constants of motion
  is itself a constant of motion, so $\{H_j, H_k\}$ is independent of time.
  Consider now this bracket at a fixed point
  $(\mathbf{x}^0, \mathbf{m}^0)
  :=(x_1^0,x_2^0,\ldots,x_n^0, m_1^0,m_2^0,\ldots m_n^0)
  \in \puresector$
  which we consider as an initial condition for the peakon flow
  $(\mathbf{x}(t), \mathbf{m}(t))$, which exists globally in time
  by Theorem~\ref{thm:globalexistence}.
  Theorem~\ref{thm:asymptotics}, which will be proved later without using
  what we are proving here, shows that the peakons scatter as
  $t\to -\infty$; more precisely,
  $m_k^2$ tends to $1/\lambda_k$, while the $x_k$ move
  apart, so that the terms $O(E_{pq})$ tend to zero.
  (It should also be possible to prove these scattering properties directly
  from the peakon ODEs, along the lines of what was done for the
  Degasperis--Procesi equation in
  \cite[Theorem~2.4]{lundmark-szmigielski-DPlong},
  but we have not done that.)
  Thus, from \eqref{eq:hksym},
  $\{H_j, H_k\}(\mathbf{x}^0, \mathbf{m}^0)
  = \{H_j, H_k\}(\mathbf{x}(t), \mathbf{m}(t))
  = \lim _{t\to -\infty} \{H_j, H_k\}(\mathbf{x}(t), \mathbf{m}(t))
  = \lim _{t\to -\infty}\{e_j, e_k\}(\mathbf{x}(t),\mathbf{m}(t))$.
  Now the Poisson brackets of these symmetric functions are given by
  linear combinations of the Poisson brackets $\{m_j, m_k\}$ with
  coefficients dependent only on the amplitudes. However, from
  \eqref{eq:poisson-structure} it is clear that
  $\{m_j, m_k\}(\mathbf{x}(t), \mathbf{m}(t)) =O(E_{pq})\to 0$,
  from which it follows that
  $\{e_j, e_k\}(\mathbf{x}(t), \mathbf{m}(t))\to 0$ as $t\to -\infty$, and
  hence $\{H_j, H_k\}(\mathbf{x}^0, \mathbf{m}^0)=0$ as required.
\end{proof}

\begin{remark}
  Since the vanishing of the Poisson bracket is a purely algebraic
  relation,
  the $H_k$ Poisson commute at each point of $R^{2n}$,
  not just in the region~$\puresector$.
\end{remark}

The $\lambda_k$, which are defined as the zeros of $A(\lambda)$, are
the eigenvalues of the inverse of the matrix $TPEP$,
since $A(\lambda) = \det(I - \lambda TPEP)$.
Another reason why we restrict our attention to the case with all $m_k>0$
is that the matrix $TPEP$ can then be shown to be oscillatory
(see Section~\ref{sec:TPEP-oscillatory} in the appendix),
which implies that its eigenvalues are positive and simple.
Consequently, the $\lambda_k$ are also positive and simple,
and for definiteness we will number them such that
\begin{equation}
  \label{eq:ordering-lambda}
  0 < \lambda_1 < \dots < \lambda_n.
\end{equation}
(For another proof that the spectrum is positive and simple,
see Theorem~\ref{thm:translated-from-neumann}.)

Turning now to $B=S(\lambda)_{21}$ and $C=S(\lambda)_{31}$,
we find from \eqref{eq:S} and \eqref{eq:jump-matrix}
that they are polynomials in $\lambda$ of degree $n-1$,
with $B(0) = M_+$ and $C(0) = M_+^2$,
where $M_+ = \sum_{k=1}^N m_k \, e^{x_k}$ as before.
This means that the two \emph{Weyl functions}
\begin{equation}
  \label{eq:weylfunctions}
  \omega(\lambda) = -\frac{B(\lambda)}{A(\lambda)}
  \qquad\text{and}\qquad
  \zeta(\lambda) = -\frac{C(\lambda)}{2A(\lambda)}
\end{equation}
are rational functions of order $O(1/\lambda)$ as $\lambda\to\infty$,
having poles at the eigen\-values~$\lambda_k$.
Let $b_k$ and $c_k$ denote the residues,
\begin{equation}
  \label{eq:weyl-parfrac}
  \omega(\lambda) = \sum_{k=1}^n \frac{b_k}{\lambda-\lambda_k},
  \qquad
  \zeta(\lambda) = \sum_{k=1}^n \frac{c_k}{\lambda-\lambda_k}.
\end{equation}
The time evolution of $(A,B,C)$, given by \eqref{eq:ABCevolution},
translates into
\begin{equation}
  \label{eq:weyl-evolution}
  \dot\omega(\lambda) = \frac{\omega(\lambda)-\omega(0)}{\lambda},
  \qquad
  \dot\zeta(\lambda) = - \omega(0) \, \dot\omega(\lambda).
\end{equation}
Comparing residues on both sides in \eqref{eq:weyl-evolution} gives
\begin{equation}
  \label{eq:residue-evolution}
  \dot b_k = \frac{b_k}{\lambda_k},
  \qquad
  \dot c_k = - \omega(0) \, \frac{b_k}{\lambda_k}
  = \sum_{m=1}^n \frac{b_m b_k}{\lambda_m \lambda_k}.
\end{equation}
This at once implies $b_k(t)=b_k(0) \, e^{t/\lambda_k}$, and
integrating $\dot c_k(\tau)$ from $\tau = -\infty$
(assuming that $c_k$ vanishes there)
to~$\tau=t$ then gives
\begin{equation}
  \label{eq:ck}
  c_k = \sum_{m=1}^n \frac{b_m b_k}{\lambda_m + \lambda_k}.
\end{equation}
A purely algebraic proof of this relation between the Weyl functions,
not relying on time dependence and integration, will be given below;
see Theorem~\ref{thm:translated-from-neumann}.
We note the identities
$\sum_1^n c_k/\lambda_k = \frac12 (\sum_1^n b_k/\lambda_k)^2$
and $\sum_1^n \lambda_k c_k = \frac12 (\sum_1^n b_k)^2$.

The multipeakon solution is obtained as follows.
The initial data $x_k(0)$, $m_k(0)$
(for $k=1,\dots,n$)
determine initial spectral data
$\lambda_k(0)$, $b_k(0)$,
which after time $t$ have evolved to
$\lambda_k(t) = \lambda_k(0)$, $b_k(t) = b_k(0) \, e^{t/\lambda_k}$
(since the $\lambda_k$ are the zeros of the time-invariant polynomial
$A(\lambda)$, and since the $b_k$ satisfy \eqref{eq:residue-evolution}).
Solving the inverse spectral problem for these spectral data at time~$t$
gives the solution $x_k(t)$, $m_k(t)$.
The remainder of the paper is devoted to this inverse spectral problem.

\section{The dual cubic string}
\label{sec:dual}

Just like for the Camassa--Holm and Degasperis--Procesi equations,
some terms in the Lax pair's spatial equation (equation
\eqref{eq:lax-x} in this case, repeated as \eqref{eq:lax-x-again}
below) can be removed by a change of both dependent and independent
variables. We refer to this as a Liouville transformation, since it is
reminiscent of the transformation used for bringing a second-order
Sturm--Liouville operator to a simple normal form.
This simplification reveals an interesting connection between the
Novikov equation and the Degasperis--Procesi equation, and allows us
to solve the inverse spectral problem by making use of the tools
developed in the study of the latter.

\begin{theorem}
  \label{thm:liouville}
  The spectral problem
  \begin{equation}
    \label{eq:lax-x-again}
    \frac{\partial}{\partial x}
    \begin{pmatrix} \psi_1 \\ \psi_2 \\ \psi_3 \end{pmatrix} =
    \begin{pmatrix}
      0 & zm(x) & 1 \\
      0 & 0 & zm(x) \\
      1 & 0 & 0
    \end{pmatrix}
    \begin{pmatrix} \psi_1 \\ \psi_2 \\ \psi_3 \end{pmatrix}
  \end{equation}
  on the real line $x\in \R$, with boundary conditions
  \begin{equation}
    \label{eq:boundary-conditions-x}
    \begin{split}
      \psi_2(x) &\to 0, \quad\text{as $x\to-\infty$},\\
      e^x \psi_3(x) &\to 0, \quad\text{as $x\to-\infty$},\\
      e^{-x} \psi_3(x) &\to 0, \quad\text{as $x\to+\infty$},
    \end{split}
  \end{equation}
  is equivalent (for $z \neq 0$), under the change of variables
  \begin{equation}
    \label{eq:liouville-trf}
    \begin{split}
      y &= \tanh x,\\
      \phi_1(y) &= \psi_1(x) \cosh x - \psi_3(x) \sinh x,\\
      \phi_2(y) &= z \, \psi_2(x),\\
      \phi_3(y) &= z^2 \, \psi_3(x) / \cosh x,\\
      g(y) &= m(x) \, \cosh^3 x,\\
      \lambda &= -z^2,
    \end{split}
  \end{equation}
  to the ``dual cubic string'' problem
  \begin{equation}
    \label{eq:dual-cubic-three-component}
    \frac{\partial}{\partial y}
    \begin{pmatrix} \phi_1 \\ \phi_2 \\ \phi_3 \end{pmatrix}
    =
    \begin{pmatrix}
      0 & g(y) & 0 \\
      0 & 0 & g(y) \\
      -\lambda & 0 & 0
    \end{pmatrix}
    \begin{pmatrix} \phi_1 \\ \phi_2 \\ \phi_3 \end{pmatrix}
  \end{equation}
  on the finite interval $-1<y<1$, with boundary conditions
  \begin{equation}
    \label{eq:boundary-conditions-y}
    \phi_2(-1) = \phi_3(-1) = 0
    \qquad
    \phi_3(1) = 0.
  \end{equation}
  In the discrete case $m = 2 \sum_{k=1}^n m_k \, \delta_{x_k}$,
  the relation between the measures $m$ and $g$ should be interpreted as
  \begin{equation}
    \label{eq:m-g-discrete}
    g(y) = \sum_{k=1}^n g_k \delta_{y_k},
    \qquad
    y_k = \tanh x_k,
    \qquad
    g_k = 2 \, m_k \cosh x_k.
  \end{equation}
\end{theorem}

\begin{proof}
  Straightforward computation using the chain rule and,
  for the discrete case,
  $\delta_{x_k} = \frac{dy}{dx}(x_k) \, \delta_{y_k}$.
\end{proof}

\begin{remark}
  The cubic string equation,
  which plays a crucial role in the derivation of the
  Degasperis--Procesi multipeakon solution
  \cite{lundmark-szmigielski-DPlong},
  is
  \begin{equation}
    \label{eq:cubicstring}
    \partial_y^3 \phi = -\lambda g \phi,
  \end{equation}
  which can be written as a system by letting
  $\Phi = (\phi_1,\phi_2,\phi_3) = (\phi,\phi_y,\phi_{yy})$:
  \begin{equation}
    \label{eq:primal-cubic-three-component}
    \frac{\partial}{\partial y}
    \begin{pmatrix} \phi_1 \\ \phi_2 \\ \phi_3 \end{pmatrix}
    =
    \begin{pmatrix}
      0 & 1 & 0 \\
      0 & 0 & 1 \\
      -\lambda g(y) & 0 & 0
    \end{pmatrix}
    \begin{pmatrix} \phi_1 \\ \phi_2 \\ \phi_3 \end{pmatrix}.
  \end{equation}
  The duality between \eqref{eq:dual-cubic-three-component} and
  \eqref{eq:primal-cubic-three-component} manifests itself in the
  discrete case as an interchange of the roles of masses $g_k$ and
  distances $l_k = y_{k+1} - y_k$; see Section~\ref{sec:neumann}.
  When the mass distribution is given by a continuous function $g(y) > 0$,
  the systems are instead related via the change of variables defined by
  \begin{equation}
    \label{eq:duality-transformation}
    \frac{d\tilde{y}}{dy} = g(y) = \frac{1}{\tilde{g}(\tilde{y})},
  \end{equation}
  where $y$ and $g(y)$ refer to the primal cubic string
  \eqref{eq:primal-cubic-three-component},
  and $\tilde{y}$ and $\tilde{g}(\tilde{y})$ to the dual cubic string
  \eqref{eq:dual-cubic-three-component}
  (or the other way around; the transformation
  \eqref{eq:duality-transformation} is obviously symmetric
  in $y$ and $\tilde{y}$,
  so that the dual of the dual is the original cubic string again).
\end{remark}

\begin{remark}
  The concept of a dual string figures prominently in the work of
  Krein on the ordinary string equation
  $\partial_y^2 \phi = -\lambda g \phi$
  (as opposed to the cubic string).
  For a comprehensive account of Krein's theory, see
  \cite{dym-mckean-gaussian}.
\end{remark}

\begin{remark}
  As a motivation for the transformation \eqref{eq:liouville-trf},
  we note that one can eliminate $\psi_1$
  from \eqref{eq:lax-x-again}, which gives
  $\partial_x \psi_2 = zm \psi_3$, $(\partial_x^2-1)\,\psi_3 = zm\psi_2$.
  From the study of Camassa--Holm peakons
  \cite{beals-sattinger-szmigielski-moment}
  it is known that the transformation
  $y=\tanh x$, $\phi(y)=\psi(x) / \cosh x$
  takes the expression $(\partial_x^2-1) \, \psi$ to a multiple
  of~$\phi_{yy}$,
  so it is not far-fetched to try something similar on $\psi_3$
  while leaving $\psi_2$ essentially unchanged.
\end{remark}

From now on we concentrate on the discrete case.
The Liouville transformation maps the piecewise defined
$(\psi_1,\psi_2,\psi_3)$ given by \eqref{piecewise-x} to
\begin{equation}
  \label{eq:piecewise-y}
  \begin{pmatrix} \phi_1 \\ \phi_2 \\ \phi_3 \end{pmatrix} =
  \begin{pmatrix}
    A_k(\lambda) - \lambda \, C_k(\lambda) \\ -2 \lambda \, B_k(\lambda) \\ -\lambda \, A_k(\lambda) \, (1+y) - \lambda^2 \, C_k(\lambda) \, (1-y)
  \end{pmatrix}
  \quad\text{for $y_k < y < y_{k+1}$}.
\end{equation}
The initial values $(A_0,B_0,C_0)=(1,0,0)$ thus correspond
to $\Phi(-1;\lambda)=(1,0,0)^t$, where
$\Phi(y;\lambda)=\bigl(\phi_1,\phi_2,\phi_3 \bigr)^t$,
and at the right endpoint $y=1$ we have
\begin{equation}
  \label{eq:right-endpoint}
  \Phi(1;\lambda) =
  \begin{pmatrix}
    A_n(\lambda) - \lambda \, C_n(\lambda) \\ -2 \lambda \, B_n(\lambda) \\ -2 \lambda \, A_n(\lambda)
  \end{pmatrix}.
\end{equation}
In particular, the condition $A_n(\lambda)=0$ defining the spectrum
corresponds to $\phi_3(1;\lambda) = 0$, except that the latter condition
gives an additional eigenvalue $\lambda_0 = 0$ which is only present
on the finite interval.
(This is not a contradiction, since the Liouville transformation from
the line to the interval is not invertible when $z=-\lambda^2=0$.)

The component $\phi_3$ is continuous and piecewise linear,
while $\phi_1$ and $\phi_2$ are piecewise constant with jumps at
the points $y_k$ where the measure $g$ is supported.
More precisely, at point mass number $k$ we have
\begin{equation}
  \label{eq:jump-relations}
  \begin{split}
    \phi_1(y_k^+) - \phi_1(y_k^-) &= g_k \avg{\phi_2(y_k)},\\
    \phi_2(y_k^+) - \phi_2(y_k^-) &= g_k \, \phi_3(y_k),
  \end{split}
\end{equation}
and in interval number $k$, with length $l_k = y_{k+1}-y_k$,
\begin{equation}
  \label{eq:interval-relations}
  \phi_3(y_{k+1}^-) - \phi_3(y_k^+)
  = l_k \, \partial_y \phi_3(y_k^+)
  = -\lambda \, l_k \, \phi_1(y_k^+).
\end{equation}
In terms of the vector $\Phi$ these relations take the form
\begin{equation}
  \label{eq:Gk}
  \Phi(y_k^+) =
  \begin{pmatrix}
    1 & g_k & \frac12 g_k^2 \\ 0 & 1 & g_k \\ 0 & 0 & 1
  \end{pmatrix}
  \Phi(y_k^-),
\end{equation}
and
\begin{equation}
  \label{eq:Lk}
  \Phi(y_{k+1}^-) =
  \begin{pmatrix}
    1 & 0 & 0 \\ 0 & 1 & 0 \\ -\lambda l_k & 0 & 1
  \end{pmatrix}
  \Phi(y_k^+),
\end{equation}
respectively.
If we introduce the notation
\begin{equation}
  \label{eq:jump-matrices}
  G(x,\lambda) =
  \begin{pmatrix}
    1 & 0 & 0 \\ 0 & 1 & 0 \\ -\lambda x & 0 & 1
  \end{pmatrix},
  \qquad
  L(x) =
  \begin{pmatrix}
    1 & x & \frac12 x^2 \\ 0 & 1 & x \\ 0 & 0 & 1
  \end{pmatrix},
\end{equation}
it follows immediately that
\begin{equation}
  \label{eq:phi-matrix-product}
  \Phi(1;\lambda) =
  G(l_n,\lambda) \,\, L(g_n) \, \dots \, G(l_2,\lambda) \,\, L(g_2) \,\, G(l_1,\lambda) \,\, L(g_1) \,\,  G(l_0,\lambda)
  \left( \begin{smallmatrix} 1 \\ 0 \\ 0 \end{smallmatrix} \right).
\end{equation}

We define the Weyl functions $W$ and $Z$ of the dual cubic string to be
\begin{equation}
  \label{eq:weylfunctions-interval}
  W(\lambda) = -\frac{\phi_2(1;\lambda)}{\phi_3(1;\lambda)},
  \qquad
  Z(\lambda) = -\frac{\phi_1(1;\lambda)}{\phi_3(1;\lambda)}.
\end{equation}
It is clear from \eqref{eq:right-endpoint} that they are related
to the Weyl functions $\omega$ and $\zeta$ previously defined on the real
line (see \eqref{eq:weylfunctions}) as follows:
\begin{equation}
  \label{eq:weyl-relations}
  \begin{split}
    W(\lambda) &= -\frac{B_n(\lambda)}{A_n(\lambda)} = \omega(\lambda)
    = \sum_{k=1}^n \frac{b_k}{\lambda-\lambda_k},
    \\
    Z(\lambda) &=
    \frac{A_n(\lambda) - \lambda C_n(\lambda)}{2 \lambda A_n(\lambda)} =
    \frac{1}{2 \lambda} + \zeta(\lambda) =
    \frac{1}{2 \lambda} + \sum_{k=1}^n \frac{c_k}{\lambda-\lambda_k}.
  \end{split}
\end{equation}

\section{Relation to the Neumann-like cubic string}
\label{sec:neumann}

Kohlenberg, Lundmark and Szmigielski \cite{kohlenberg-lundmark-szmigielski}
studied the discrete cubic string with Neumann-like boundary conditions.
We will briefly recall some results from that paper,
with notation and sign conventions slightly altered to suit our needs here.
The spectral problem in question is
\begin{equation}
  \label{eq:cubic-neumann}
  \begin{split}
    \phi_{yyy}(y) &= -\lambda g(y)\phi(y) \quad\text{for $y\in\R$},\\
    \phi_y(-\infty) &= \phi_{yy}(-\infty) = 0,
    \qquad
    \phi_{yy}(\infty) = 0,
  \end{split}
\end{equation}
where $g = \sum_{k=0}^n g_k \, \delta_{y_k}$ is a discrete measure
with $n+1$ point masses $g_0,\dots,g_n$ at positions $y_0 < y_1 < \dots < y_n$;
between these points are $n$ finite intervals of length
$l_1,\dots,l_n$ (where $l_k = y_k-y_{k-1}$).
Since $\phi_{yyy} = 0$ away from the point masses,
the boundary conditions can equally well be written as
\begin{equation*}
    \phi_y(y_0^-) = \phi_{yy}(y_0^-) = 0,
    \qquad
    \phi_{yy}(y_n^+) = 0.
\end{equation*}
Using the normalization $\phi(-\infty)=1$ (or $\phi(y_0^-)=1$) and the notation
$\Phi = (\phi,\phi_y,\phi_{yy})^t$, one finds
\begin{equation}
  \label{eq:phi-matrix-product-neumann}
  \Phi(y_n^+;\lambda) =
  G(g_n,\lambda) \,\, L(l_n) \, \dots \, G(g_2,\lambda) \,\, L(l_2) \,\, G(g_1,\lambda) \,\, L(l_1) \,\,  G(g_0,\lambda)
  \left( \begin{smallmatrix} 1 \\ 0 \\ 0 \end{smallmatrix} \right),
\end{equation}
with matrices $G$ and $L$ as in \eqref{eq:jump-matrices} above.
Under the assumption that all $g_k > 0$,
the zeros of $\phi_{yy}(y_n^+;\lambda)$,
which constitute the spectrum,
are
\begin{equation*}
  0 = \lambda_0 < \lambda_1 < \dots < \lambda_n,
\end{equation*}
and the Weyl functions are
\begin{equation}
  \label{eq:weyl-neumann}
  \begin{split}
    W(\lambda) &= -\frac{\phi_y(y_n^+;\lambda)}{\phi_{yy}(y_n^+;\lambda)}
    = \sum_{k=1}^n \frac{b_k}{\lambda-\lambda_k},
    \\
    Z(\lambda) &= -\frac{\phi(y_n^+;\lambda)}{\phi_{yy}(y_n^+;\lambda)}
    = \frac{1}{\gamma \lambda} + \sum_{k=1}^n \frac{c_k}{\lambda-\lambda_k},
    \qquad \gamma = \sum_{k=0}^n g_k,
  \end{split}
\end{equation}
with all $b_k > 0$. They satisfy the identity
\begin{equation}
  \label{eq:ZWidentity}
  Z(\lambda) + Z(-\lambda) + W(\lambda) W(-\lambda) = 0,
\end{equation}
from which it follows, by taking the residue at $\lambda=\lambda_k$, that
\begin{equation}
  \label{eq:ck-again}
  c_k = \sum_{m=1}^n \frac{b_m b_k}{\lambda_m + \lambda_k}.
\end{equation}
Thus $Z(\lambda)$ is uniquely determined by the function $W(\lambda)$
and the constant~$\gamma$.

Now note that \eqref{eq:phi-matrix-product-neumann}
is exactly the same kind of relation as \eqref{eq:phi-matrix-product},
except that the roles of $g_k$ and $l_k$ are interchanged,
and the right endpoint is called $y=y_n^+$ instead of $y=1$.
The definitions of the Weyl functions \eqref{eq:weyl-neumann} also
correspond perfectly to the Weyl functions \eqref{eq:weylfunctions-interval}
for the dual cubic string.
Therefore, all the results above are also true in the setting of
the dual cubic string.
The assumption that the $n$ distances $l_k$ and the $n+1$ point masses $g_k$
are all positive for the Neumann cubic string corresponds of course to
the requirement that the $n$ point masses $g_k$ and the $n+1$ distances $l_k$
are positive for the dual cubic string.
The constant $\gamma = \sum_{k=0}^n g_k$ in the term $1/\gamma\lambda$
in \eqref{eq:weyl-neumann} corresponds
to the constant $2$ in the term $1/2\lambda$ in \eqref{eq:weyl-relations},
since $\sum_{k=0}^n l_k = 2$ is the length of the interval $-1<y<1$.
In summary:

\begin{theorem}
  \label{thm:translated-from-neumann}
  Assume that all point masses $g_k$ are positive.
  Then the discrete dual cubic string of Theorem~\ref{thm:liouville}
  has nonnegative and simple spectrum, with eigenvalues
  $0 = \lambda_0 < \lambda_1 < \dots < \lambda_n$,
  and its Weyl functions \eqref{eq:weylfunctions-interval} have positive
  residues and satisfy \eqref{eq:ZWidentity} and \eqref{eq:ck-again}.
  In particular, the second Weyl function~$Z(\lambda)$ is uniquely determined
  by the first Weyl function~$W(\lambda)$.
\end{theorem}

\section{Inverse spectral problem}
\label{sec:inverse}

The inverse spectral problem for the discrete dual cubic string
consists in recovering the positions and masses
$\left\{ y_k, g_k \right\}_{k=1}^n$
given the spectral data consisting of eigenvalues and residues
$\left\{ \lambda_k, b_k \right\}_{k=1}^n$
(or, equivalently, given the first Weyl function $W(\lambda)$).
The corresponding problem for the Neumann-like cubic string was
solved in \cite{kohlenberg-lundmark-szmigielski},
and we need only translate the results, as in Section~\ref{sec:neumann}.
See also \cite{lundmark-szmigielski-DPlong}
for more information about inverse problems of this kind and
\cite{bertola-gekhtman-smigielski-peakons-cauchy} for the underlying
theory of Cauchy biorthogonal polynomials.

To begin with, we state the result in terms of the bimoment determinants
$\detD{ab}{m}$ and~$\detDprime{m}$ defined below.
Things will become more explicit in the next section
(Corollary~\ref{cor:inverse-spectral-dual-cubic}), where the determinants
are expressed directly in terms of the $\lambda_k$ and $b_k$.

\begin{definition}
  Suppose $\mu$ is a measure on $\R_+$ (the positive part of the real line)
  such that its moments,
  \begin{equation}
    \label{eq:betaa}
    \beta_a=\int \kappa ^a \, d\mu(\kappa),
  \end{equation}
  and its bimoments with respect to the Cauchy kernel $K(x,y)=1/(x+y)$,
  \begin{equation}
    \label{eq:Iab}
    I_{ab} = I_{ba} = \iint \frac{\kappa^a \, \lambda^b}{\kappa+\lambda} \, d\mu(\kappa) \, d\mu(\lambda),
  \end{equation}
  are finite. For $m \ge 1$,
  let $\detD{ab}{m}$ denote the determinant of the $m \times m$ bimoment
  matrix which starts with $I_{ab}$ in the upper left corner:
  \begin{equation}
    \label{eq:Dabm}
    \detD{ab}{m} =
    \begin{vmatrix}
      I_{ab} & I_{a,b+1} & \dots & I_{a,b+m-1} \\
      I_{a+1,b} & I_{a+1,b+1} & \dots & I_{a+1,b+m-1} \\
      I_{a+2,b} & I_{a+2,b+1} & \dots & I_{a+2,b+m-1} \\
      \vdots &&& \vdots \\
      I_{a+m-1,b} & I_{a+m-1,b+1} & \dots & I_{a+m-1,b+m-1} \\
    \end{vmatrix}
    = \detD{ba}{m}.
  \end{equation}
  Let $\detD{ab}{0} = 1$, and $\detD{ab}{m} = 0$ for $m < 0$.

  Similarly, for $m \ge 2$,
  let $\detDprime{m}$ denote the $m \times m$ determinant
  \begin{equation}
    \label{eq:D'abm}
    \detDprime{m} =
    \begin{vmatrix}
     \beta_0 & I_{10} & I_{11} & \dots & I_{1,m-2} \\
     \beta_1 & I_{20} & I_{21} & \dots & I_{2,m-2} \\
     \beta_2 & I_{30} & I_{31} & \dots & I_{3,m-2} \\
     \vdots &&&& \vdots \\
     \beta_{m-1} & I_{m0} & I_{m1} & \dots & I_{m,m-2} \\
    \end{vmatrix},
  \end{equation}
  and define $\detDprime{1} = \beta_0$
  and $\detDprime{m} = 0$ for $m < 1$.
\end{definition}

\begin{theorem}
  \label{thm:inverse-spectral-dual-cubic}
  Given constants $0< \lambda_1 < \dots < \lambda_n$ and
  $b_1,\dots,b_n > 0$,
  define the spectral measure
  \begin{equation}
    \label{eq:spectral-measure}
    \mu = \sum_{i=1}^n b_i \, \delta_{\lambda_i},
  \end{equation}
  and let $I_{ab}$ be its bimoments,
  \begin{equation}
    I_{ab} = \iint \frac{\kappa^a \, \lambda^b}{\kappa+\lambda} \, d\mu(\kappa) \, d\mu(\lambda)
    = \sum_{i=1}^n \sum_{j=1}^n \frac{\lambda_i^a \lambda_j^b}{\lambda_i + \lambda_j}\,b_i b_j.
  \end{equation}
  Then the unique discrete dual cubic string (with positive masses $g_k$)
  having the Weyl function
  \begin{equation*}
    W(\lambda) = \sum_{k=1}^n \frac{b_k}{\lambda-\lambda_k}
    = \int\frac{d\mu(\kappa)}{\lambda-\kappa}
  \end{equation*}
  is given by
  \begin{equation}
    \label{eq:ykgk-bimoment}
    y_{k'} = \frac{\detD{00}{k} - \frac12 \detD{11}{k-1}}{\detD{00}{k} + \frac12 \detD{11}{k-1}},
    \qquad
    g_{k'} = 2 \frac{\detD{00}{k} + \frac12 \detD{11}{k-1}}{\detDprime{k}},
  \end{equation}
  where $k'=n+1-k$ for $k=0,\dots,n+1$.
  The distances between the masses are given by
  \begin{equation}
    \label{eq:lk-bimoment}
    l_{k'-1} = y_{k'} - y_{k'-1}
    = \frac{\Bigl( \detD{10}{k} \Bigr)^2}{\Bigl( \detD{00}{k} + \frac12 \detD{11}{k-1} \Bigr) \Bigl( \detD{00}{k+1} + \frac12 \detD{11}{k} \Bigr)}.
  \end{equation}
\end{theorem}

\begin{proof}
  For $0\le k \le n$, let $a^{(2k+1)}(\lambda)$
  be the product of the first $2k+1$ factors in \eqref{eq:phi-matrix-product},
  \begin{multline}
    \label{eq:a}
      a^{(2k+1)}(\lambda) =
      G(l_n,\lambda) \,\, L(g_n) \,\, G(l_{n-1},\lambda) \,\, L(g_{n-1}) \, \dots
      \\
      \dots \, G(l_{k'},\lambda) \,\, L(g_{k'}) \,\, G(l_{k'-1},\lambda),
  \end{multline}
  where $k'=n+1-k$.
  By Lemma~4.1 and Theorem~4.2 in
  \cite{kohlenberg-lundmark-szmigielski},
  the entries in the first column of $a = a^{(2k+1)}(\lambda)$,
  \begin{equation*}
    \begin{pmatrix} a_{11} \\ a_{21} \\ a_{31} \end{pmatrix}
    =:
    \begin{pmatrix} \widehat{P} \\ P \\ Q \end{pmatrix},
  \end{equation*}
  satisfy what in \cite{kohlenberg-lundmark-szmigielski} was called
  a ``Type I'' approximation problem.
  This means that $(\widehat{P}(\lambda),P(\lambda),Q(\lambda))$
  are polynomials in $\lambda$ of degree $k$, $k$, $k+1$, respectively,
  satisfying the normalization conditions
  \begin{equation*}
    \widehat{P}(0)=1,\qquad
    P(0)=0, \qquad
    Q(0)=0,
  \end{equation*}
  the approximation conditions
  \begin{equation*}
    Q(\lambda) W(\lambda) + P(\lambda) = O(1),\qquad
    Q(\lambda) Z(\lambda) + \widehat{P}(\lambda) = O(\lambda^{-1}),
    \qquad
    \text{as $\lambda\to\infty$},
  \end{equation*}
  and the symmetry condition
  \begin{equation*}
    Q(\lambda) Z(-\lambda) - P(\lambda) W(-\lambda) - \widehat{P}(\lambda)
    = O(\lambda^{-k-1}),
    \qquad
    \text{as $\lambda\to\infty$}.
  \end{equation*}
  According to Theorem~4.15 in \cite{kohlenberg-lundmark-szmigielski},
  this determines $(\widehat{P},P,Q)$ uniquely; in particular,
  the coefficients of
  $a^{(2k+1)}_{31}(\lambda) = Q(\lambda) = \sum_{i=1}^{k+1} q_i \lambda^i$
  are given by the nonsingular linear system
  \begin{equation}
    \label{eq:system-for-Q-explicit}
    \begin{pmatrix}
      I_{00}+\frac{1}{2} & I_{01} & \cdots & I_{0k}\\
      I_{10} & I_{11} & \cdots & I_{1k} \\
      I_{20} & I_{21} & \cdots & I_{2k} \\
      \vdots &&& \vdots \\
      I_{k0} & I_{k1} & \cdots & I_{kk}
    \end{pmatrix}
    \begin{pmatrix}
      q_1 \\ q_2 \\q_3\\ \vdots \\ q_{k+1}
    \end{pmatrix}
    =
    -\begin{pmatrix}
      1\\ 0\\0\\\vdots \\ 0
    \end{pmatrix}.
  \end{equation}
  From \eqref{eq:a} one finds that
  \begin{equation}
    \label{eq:entries-coeffs}
    \begin{split}
      a^{(2k+1)}_{31}(\lambda) &= (-\lambda) (l_n + l_{n-1} + \dots + l_{k'-1})
      + \dots
      \\ & \quad + (-\lambda)^{k+1} \left(
        \frac{g_n^2}{2} \frac{g_{n-1}^2}{2} \dots \frac{g_{k'}^2}{2}
        l_n l_{n-1} \dots l_{k'-1}
      \right),
    \end{split}
  \end{equation}
  and the lowest and highest coefficients are then extracted from
  \eqref{eq:system-for-Q-explicit} using Cramer's rule:
  \begin{equation}
    \label{eq:q-bottom-top}
    \begin{split}
      -q_1 &
      = \frac{\detD{11}{k}}{\detD{00}{k+1} + \frac12 \detD{11}{k}}
      = \sum_{j=k'-1}^n \!\!\! l_j
      = 1-y_{k'-1},
      \\
      (-1)^{k+1} q_{k+1} &
      = \frac{\detD{10}{k}}{\detD{00}{k+1} + \frac12 \detD{11}{k}}
      = \left( \prod_{j=k'}^n \frac{g_j^2 \, l_j}{2} \right) l_{k'-1}.
    \end{split}
  \end{equation}
  The first equation gives a formula for $y_{k'-1}$ right away,
  and of course also for $y_{k'}$ (with $1\le k \le n+1$) after renumbering.
  This formula \eqref{eq:ykgk-bimoment} for $y_{k'}$ holds also for $k=0$,
  since it gives $y_{0'} = y_{n+1} = +1$ because of the way
  $\detD{ab}{m}$ is defined for $m\le 0$.
  (That it indeed gives $y_{(n+1)'} = y_0 = -1$ when $k=n+1$ is not as obvious;
  this depends on $\detD{00}{n+1}$ being zero when the measure $\mu$
  is supported on only $n$ points.
  See \cite[Appendix~B]{kohlenberg-lundmark-szmigielski}.)
  Subtraction gives a formula for $l_{k'-1}$ which simplifies to
  \eqref{eq:lk-bimoment} with the help of ``Lewis Carroll's identity''
  \cite[Prop.~10]{krattenthaler} applied to the determinant
  $\detD{00}{k+1}$:
  \begin{equation}
    \label{eq:LewisCarroll}
    \detD{00}{k+1} \detD{11}{k-1}
    = \detD{00}{k} \detD{11}{k} - \detD{10}{k} \detD{01}{k}.
  \end{equation}
  Finally, the second formula in \eqref{eq:q-bottom-top},
  divided by the corresponding formula with $k$ replaced by \mbox{$k-1$},
  gives an expression for $\frac12 \, g_{k'}^2 \, l_{k'-1}$
  from which one obtains
  \begin{equation*}
    g_{k'} = \Bigl( \detD{00}{k} + \frac12 \detD{11}{k-1} \Bigr)
    \sqrt{\frac{2}{\detD{10}{k} \detD{10}{k-1}}}.
  \end{equation*}
  The formula for $g_{k'}$ presented in \eqref{eq:ykgk-bimoment}
  now follows from the identity
  $(\detDprime{k})^2 = 2 \detD{10}{k} \detD{10}{k-1}$
  and the positivity of $\detDprime{k}$,
  which are immediate consequences of \eqref{eq:heine} below.
  (The determinant identity can also be proved directly
  by expanding $\detDprime{k}$ along the first column,
  squaring, and using $\beta_i \beta_j = I_{i+1,j} + I_{i,j+1}$.)
\end{proof}

\begin{remark}
  We take this opportunity to correct a couple of mistakes in
  \cite{kohlenberg-lundmark-szmigielski}:
  the formula in Corollary~4.17 should read
  $[Q_{3k+2}] = (-1)^{k+1} \mathcal{D}_k / \mathcal{A}_{k+1}$,
  and consequently it should be
  $\displaystyle m_{n-k}=\frac{\mathcal{D}^2_k}{2\mathcal{A}_{k+1}\mathcal{A}_k}$
  in (4.54).
\end{remark}

\section{Evaluation of bimoment determinants}
\label{sec:determinants}

The aim of this section is just to state some formulas for the bimoment
determinants $\detD{ab}{m}$ and $\detDprime{m}$, taken from
\cite[Lemma~4.10]{lundmark-szmigielski-DPlong}
and
\cite[Appendix~B]{kohlenberg-lundmark-szmigielski}.
Quite a lot of notation is needed.

\begin{definition}
  \label{def:symmfcn}
  For $k \ge 1$, let
  \begin{equation}
    \label{eq:integrals-tuv}
    \begin{split}
      t_k &= \frac{1}{k!} \int_{\R^k} \frac{\Delta (x)^2}{\Gamma(x)} \,
      \frac{d\mu^k(x)}{x_1 x_2 \ldots x_k},
      \\
      u_k &=
      \frac{1}{k!}
      \int_{\R^k} \frac{\Delta (x)^2}{\Gamma(x)} d\mu^k(x),
      \\
      v_k &=
      \frac{1}{k!}
      \int_{\R^k} \frac{\Delta(x)^2}{\Gamma(x)} \, x_1 x_2 \ldots x_k \, d\mu^k(x),
    \end{split}
  \end{equation}
  where
  \begin{equation}
    \begin{split}
      \Delta(x)&=\Delta(x_1,\ldots,x_k)=\prod_{i<j}(x_i-x_j),\\
      \Gamma(x)&=\Gamma(x_1,\ldots,x_k)=\prod_{i<j}(x_i+x_j).
    \end{split}
  \end{equation}
  (When $k=0$ or~$1$, let $\Delta(x)=\Gamma(x)=1$.)
  Also let $t_0=u_0=v_0=1$, and $t_k=u_k=v_k=0$ for $k<0$.
\end{definition}

When $\mu = \sum_{k=1}^n b_k \, \delta_{\lambda_k}$, the integrals
$t_k$, $u_k$, $v_k$ reduce to the sums $T_k$, $U_k$, $V_k$ below.

\begin{definition}
  \label{def:notation-galore}
  For $k\ge 0$, let $\binom{[1,n]}{k}$ denote the set of $k$-element subsets
  $I=\{ i_1<\dots<i_k \}$ of the integer interval $[1,n]=\{ 1,\dots,n \}$.
  For $I\in\binom{[1,n]}{k}$, let
  \begin{equation}
    \Delta_I=\Delta(\lambda_{i_1},\dots,\lambda_{i_k}),
    \qquad
    \Gamma_I=\Gamma(\lambda_{i_1},\dots,\lambda_{i_k}),
  \end{equation}
  with the special cases
  $\Delta_{\emptyset}=\Gamma_{\emptyset}=\Delta_{\{i\}}=\Gamma_{\{i\}}=1$.
  Furthermore, let
  \begin{equation*}
    \lambda_I = \prod_{i\in I} \lambda_i,
    \qquad
    b_I = \prod_{i\in I} b_i,
  \end{equation*}
  with $\lambda_{\emptyset}=b_{\emptyset}=1$.
  Using the abbreviation $\Psi_I = \ds\frac{\Delta_I^2}{\Gamma_I}$, let
  \begin{equation}
    T_k = \sum_{I\in\binom{[1,n]}{k}} \frac{\Psi_I b_I}{\lambda_I},
    \qquad
    U_k = \sum_{I\in\binom{[1,n]}{k}} \Psi_I b_I,
    \qquad
    V_k = \sum_{I\in\binom{[1,n]}{k}} \Psi_I \lambda_I b_I,
  \end{equation}
  and
  \begin{equation}
    \begin{split}
      W_k &=
      \begin{vmatrix}
        U_k & V_{k-1} \\
        U_{k+1} & V_k
      \end{vmatrix}
      = U_k V_k - U_{k+1} V_{k-1},
      \\
      Z_k &=
      \begin{vmatrix}
        T_k & U_{k-1} \\
        T_{k+1} & U_k
      \end{vmatrix}
      = T_k U_k - T_{k+1} U_{k-1}.
    \end{split}
  \end{equation}
  (To be explicit, $U_0=V_0=T_0=1$, and $U_k=V_k=T_k=0$ for $k<0$ or $k>n$.)
\end{definition}

We can now finally state the promised formulas for the bimoment determinants.

\begin{lemma}
  For all $m$,
  \begin{equation}
    \label{eq:heine}
    \begin{aligned}
      \detD{00}{m} &= \frac{\begin{vmatrix} t_m & u_{m-1} \\ t_{m+1} & u_m \end{vmatrix}}{2^m},
      &
      \detD{11}{m} &= \frac{\begin{vmatrix} u_m & v_{m-1} \\ u_{m+1} & v_m \end{vmatrix}}{2^m},
      \\
      \detD{10}{m} &= \frac{\left( u_m \right)^2}{2^m},
      &
      \detDprime{m} &= \frac{u_m u_{m-1}}{2^{m-1}}.
    \end{aligned}
  \end{equation}
  In the discrete case when $\ds\mu = \sum_{k=1}^n b_k \, \delta_{\lambda_k}$,
  this reduces to
  \begin{equation}
    \label{eq:heine2}
    \detD{00}{m} = \frac{Z_m}{2^m},
    \quad
    \detD{11}{m} = \frac{W_m}{2^m},
    \quad
    \detD{10}{m} = \frac{\left( U_m \right)^2}{2^m},
    \quad
    \detDprime{m} = \frac{U_m U_{m-1}}{2^{m-1}}.
  \end{equation}
\end{lemma}

\begin{corollary}
  \label{cor:inverse-spectral-dual-cubic}
  The solution to the inverse spectral problem for the discrete dual cubic string
  (Theorem~\ref{thm:inverse-spectral-dual-cubic}) can be expressed as
  \begin{gather}
    \label{eq:ykgk-sums}
    y_{k'} = \frac{Z_k - W_{k-1}}{Z_k + W_{k-1}},
    \qquad
    g_{k'} = \frac{Z_k + W_{k-1}}{U_k U_{k-1}},
    \\
    \label{eq:lk-sums}
    l_{k'-1} = y_{k'} - y_{k'-1}
    = \frac{2 \left( U_k \right)^4}{(Z_k + W_{k-1})(Z_{k+1} + W_k)}.
  \end{gather}
\end{corollary}

The expression $W_k$ can be evaluated explicitly in terms of $\lambda_k$
and $b_k$, although the formula is somewhat involved
\cite[Lemma~2.20]{lundmark-szmigielski-DPlong}:
\begin{equation}
  \label{eq:Wk-explicit}
  \begin{split}
    W_k
    &= \sum_{I\in\binom{[1,n]}{k}} \frac{\Delta_I^4}{\Gamma_I^2} \lambda_I b_I^2
    \\
    &+
    \sum_{m=1}^k
    \sum_{\substack{I\in\binom{[1,n]}{k-m} \\
        J\in\binom{[1,n]}{2m} \\
        I \cap J = \emptyset}}
    b_I^2 b_J
    \Biggl\{
    2^{m+1}
    \frac{\Delta_I^4 \Delta_{I,J}^2 \lambda_{I \cup J}}{\Gamma_I \, \Gamma_{I \cup J}}
    \Biggl(
    \!\!\!
    \sum_{\substack{C \cup D = J \\ \abs{C}=\abs{D}=m \\ \min(C)<\min(D)}} \Delta_C^2 \Delta_D^2 \Gamma_C \Gamma_D
    \Biggr)
    \Biggr\},
  \end{split}
\end{equation}
where
$\Delta^2_{I,J}=\ds\prod_{i\in I, j\in J}(\lambda_{i}-\lambda_{j})^2$.
The corresponding formula for $Z_k$ is obtained by replacing $b_i$ with
$b_i/\lambda_i$ everywhere.

\section{The multipeakon solution}

In order to obtain the solution to the inverse spectral problem on
the real line, which provides the multipeakon solution,
we merely have to map the formulas for the interval
(Corollary~\ref{cor:inverse-spectral-dual-cubic})
back to the line via the Liouville transformation \eqref{eq:m-g-discrete}.

We remind the reader that in this paper we primarily study the pure
peakon case where it is assumed that all $m_k>0$ and also that
$x_1<\dots<x_n$.
This assumption guarantees that the solutions are globally defined in
time (Theorem~\ref{thm:globalexistence}) and, regarding the spectral data,
that all $b_k>0$ and $0 < \lambda_1 < \dots < \lambda_n$
(Theorem~\ref{thm:translated-from-neumann}).
Details regarding mixed peakon-antipeakon solutions are left for
future research, but we point out that since the velocity
$\dot x_k = u(x_k)^2$ is always nonnegative,
Novikov antipeakons move to the \emph{right} just like peakons
(unlike the $b$-family \eqref{eq:b-family},
where pure peakons move to the right and antipeakons to the left,
if they are sufficiently far apart).
Nevertheless, peakons and antipeakons may collide after finite time
also for the Novikov equation,
causing division by zero in the solution formula for $m_k$ in
\eqref{eq:n-peakon-solution} below,
and this breakdown leads to the usual subtle questions regarding
continuation of the solution beyond the collision.

\begin{theorem}
\label{thm:explicit_peakons}
  In the notation of Section~\ref{sec:determinants},
  the $n$-peakon solution of Novikov's equation is given by
  \begin{equation}
    \label{eq:n-peakon-solution}
    x_{k'} = \frac12 \ln\frac{Z_k}{W_{k-1}},
    \qquad
    m_{k'} = \frac{\ds\sqrt{Z_k W_{k-1}}}{U_k U_{k-1}},
  \end{equation}
  where $k' = n+1-k$ for $k=1,\dots,n$,
  and where the time evolution is given by
  \begin{equation}
    b_k(t) = b_k(0) \, e^{t/\lambda_k}.
  \end{equation}
\end{theorem}

\begin{proof}
  The inverse of the coordinate transformation \eqref{eq:m-g-discrete} is
  \begin{equation*}
    x_k = \frac12 \ln\frac{1+y_k}{1-y_k},
    \qquad
    m_k = \frac{g_k \sqrt{1-y_k^2}}{2},
  \end{equation*}
  which upon inserting \eqref{eq:ykgk-sums} gives
  \eqref{eq:n-peakon-solution} at once.
  The evolution of $b_k$ comes from equation \eqref{eq:residue-evolution}.
\end{proof}

\begin{example}
  The two-peakon solution is
  \begin{equation}
    \label{eq:twopeakon}
    \begin{split}
      x_1 = \frac12 \ln\frac{Z_2}{W_1} &= \frac12 \ln \frac{\ds \frac{(\lambda_1-\lambda_2)^4}{(\lambda_1+\lambda_2)^2 \lambda_1 \lambda_2} \, b_1^2 b_2^2}{\ds \lambda_1 \, b_1^2 + \lambda_2 \, b_2^2 + \frac{4 \, \lambda_1 \lambda_2}{\lambda_1+\lambda_2} \, b_1 b_2}, \\
      x_2 = \frac12 \ln\frac{Z_1}{W_0} &= \frac12 \ln \left( \frac{b_1^2}{\lambda_1} + \frac{b_2^2}{\lambda_2} + \frac{4}{\lambda_1+\lambda_2} \, b_1 b_2 \right), \\
      m_1 = \frac{\ds\sqrt{Z_2 W_1}}{U_2 U_1} &= \frac{\ds \left[ \frac{(\lambda_1 - \lambda_2)^4 \, b_1^2 b_2^2}{(\lambda_1 + \lambda_2)^2 \lambda_1 \lambda_2} \left( \lambda_1 \, b_1^2 + \lambda_2 \, b_2^2 + \frac{4 \, \lambda_1 \lambda_2}{\lambda_1+\lambda_2} \, b_1 b_2 \right) \right]^{1/2}}{\ds \frac{(\lambda_1 - \lambda_2)^2 \, b_1 b_2}{\lambda_1 + \lambda_2} \, (b_1+b_2)} \\
      &= \frac{\ds \left( \lambda_1 \, b_1^2 + \lambda_2 \, b_2^2 + \frac{4 \, \lambda_1 \lambda_2}{\lambda_1+\lambda_2} \, b_1 b_2 \right)^{1/2}}{\ds \sqrt{\lambda_1 \lambda_2} \, (b_1+b_2)}, \\
      m_2 = \frac{\ds\sqrt{Z_1 W_0}}{U_1 U_0} &= \frac{\left( \ds \frac{b_1^2}{\lambda_1} + \frac{b_2^2}{\lambda_2} + \frac{4}{\lambda_1+\lambda_2} \, b_1 b_2 \right)^{1/2}}{b_1+b_2},
    \end{split}
  \end{equation}
  where the simpler of the two expressions for $m_1$ is obtained under the
  assumption that all spectral data are positive,
  and therefore only can be trusted in the pure peakon case.
  This way of writing the solution is simpler and more explicit
  than that found in \cite{hone-wang-cubic-nonlinearity}.
  In order to translate \eqref{eq:twopeakon} to the notation used there,
  write $(q_k,p_k)$ instead of $(x_k,m_k)$,
  $c_k$ instead of $1/\lambda_k$,
  and $t_0$ instead of
  $(\lambda_1^{-1} - \lambda_2^{-1})^{-1} \ln \frac{b_2(0)}{b_2(0)}$;
  then $\tanh T = (b_1-b_2)/(b_1+b_2)$
  and $\cosh^{-2} T = 4 b_1 b_2/(b_1+b_2)^2$,
  where $T = \frac12(c_1-c_2)(t-t_0)$.
\end{example}

\begin{example}
  The three-peakon solution is
  \begin{equation}
    \label{eq:threepeakon}
    \begin{aligned}
      x_1 &= \frac12 \ln \frac{Z_3}{W_2},
      & x_2 &= \frac12 \ln \frac{Z_2}{W_1},
      & x_3 &= \frac12 \ln \frac{Z_1}{W_0},
      \\
      m_1 &= \frac{\ds\sqrt{Z_3 W_2}}{U_3 U_2},
      & m_2 &= \frac{\ds\sqrt{Z_2 W_1}}{U_2 U_1},
      & m_3 &= \frac{\ds\sqrt{Z_1 W_0}}{U_1 U_0},
    \end{aligned}
  \end{equation}
  where $U_0=W_0=1$,
  \begin{equation}
    \begin{split}
      U_1 &= b_1+b_2+b_3,
      \\
      U_2 &= \Psi_{12} \, b_1 b_2 + \Psi_{13} \, b_1 b_3 + \Psi_{23} \, b_2 b_3,
      \\
      U_3 &= \Psi_{123} \, b_1 b_2 b_3,
    \end{split}
  \end{equation}
  \begin{equation}
    \begin{split}
      W_1 &=
      \lambda_1 \, b_1^2 + \lambda_2 \, b_2^2 + \lambda_3 \, b_3^2
      \\
      & \quad +
      \frac{4 \, \lambda_1 \lambda_2}{\lambda_1+\lambda_2} \, b_1 b_2 +
      \frac{4 \, \lambda_1 \lambda_3}{\lambda_1+\lambda_3} \, b_1 b_3 +
      \frac{4 \, \lambda_2 \lambda_3}{\lambda_2+\lambda_3} \, b_2 b_3,
      \\[1ex]
      W_2 &= \Psi_{12}^2 \, \lambda_1 \lambda_2 \, b_1^2 b_2^2
      + \Psi_{13}^2 \, \lambda_1 \lambda_3 \, b_1^2 b_3^2
      + \Psi_{23}^2 \, \lambda_2 \lambda_3 \, b_2^2 b_3^2
      \\
      & \quad
      + \frac{4 \, \Psi_{13} \Psi_{23} \, \lambda_1 \lambda_2 \lambda_3}{\lambda_1+\lambda_2} \, b_1 b_2 b_3^2
      + \frac{4 \, \Psi_{12} \Psi_{23} \, \lambda_1 \lambda_2 \lambda_3}{\lambda_1+\lambda_3} \, b_1 b_2^2 b_3
      \\
      & \quad
      + \frac{4 \, \Psi_{12} \Psi_{13} \, \lambda_1 \lambda_2 \lambda_3}{\lambda_2+\lambda_3} \, b_1^2 b_2 b_3,
    \end{split}
  \end{equation}
  \begin{equation}
    \begin{split}
      Z_1 &=
      \frac{b_1^2}{\lambda_1} + \frac{b_2^2}{\lambda_2} + \frac{b_3^2}{\lambda_3} +
      \frac{4}{\lambda_1+\lambda_2} \, b_1 b_2 +
      \frac{4}{\lambda_1+\lambda_3} \, b_1 b_3 +
      \frac{4}{\lambda_2+\lambda_3} \, b_2 b_3,
      \\[1ex]
      Z_2 &= \frac{\Psi_{12}^2}{\lambda_1 \lambda_2} \, b_1^2 b_2^2
      + \frac{\Psi_{13}^2}{\lambda_1 \lambda_3} \, b_1^2 b_3^2
      + \frac{\Psi_{23}^2}{\lambda_2 \lambda_3} \, b_2^2 b_3^2
      \\
      & \quad
      + \frac{4 \, \Psi_{13} \Psi_{23}}{(\lambda_1+\lambda_2) \lambda_3} \, b_1 b_2 b_3^2
      + \frac{4 \, \Psi_{12} \Psi_{23}}{(\lambda_1+\lambda_3) \lambda_2} \, b_1 b_2^2 b_3
      + \frac{4 \, \Psi_{12} \Psi_{13}}{(\lambda_2+\lambda_3) \lambda_1} \, b_1^2 b_2 b_3,
      \\[1ex]
      Z_3 &=
      \frac{\Psi_{123}^2}{\lambda_1 \lambda_2 \lambda_3} \, b_1^2 b_2^2 b_3^2,
    \end{split}
  \end{equation}
  and
  \begin{equation}
    \begin{gathered}
      \Psi_{12} = \frac{(\lambda_1 - \lambda_2)^2}{\lambda_1 + \lambda_2},
      \quad
      \Psi_{13} = \frac{(\lambda_1 - \lambda_3)^2}{\lambda_1 + \lambda_3},
      \quad
      \Psi_{23} = \frac{(\lambda_2 - \lambda_3)^2}{\lambda_2 + \lambda_3},
      \\
      \Psi_{123} = \frac{(\lambda_1 - \lambda_2)^2 (\lambda_1 - \lambda_3)^2 (\lambda_2 - \lambda_3)^2}{(\lambda_1 + \lambda_2) (\lambda_1 + \lambda_3) (\lambda_2 + \lambda_3)}.
    \end{gathered}
  \end{equation}
\end{example}

\begin{theorem}[Asymptotics]
\label{thm:asymptotics}
  Let the eigenvalues be numbered so that $0 < \lambda_1 < \dots < \lambda_n$.
  Then
  \begin{equation}
    \label{eq:asymptotics-x}
    \begin{aligned}
      \displaystyle
      x_k(t) &\sim
      \frac{t}{\lambda_k} + \log b_k(0)
      - \frac12 \ln \lambda_k
      + \sum_{i=k+1}^n
      \ln \frac{(\lambda_i-\lambda_k)^2}{(\lambda_i+\lambda_k) \lambda_i},
      && \text{as $t\to-\infty$},
      \\
      x_{k'}(t) &\sim
      \displaystyle
      \frac{t}{\lambda_k} + \log b_k(0)
      - \frac12 \ln \lambda_k
      + \sum_{i=1}^{k-1}
      \ln \frac{(\lambda_i-\lambda_k)^2}{(\lambda_i+\lambda_k) \lambda_i},
      && \text{as $t\to+\infty$},
    \end{aligned}
  \end{equation}
  where $k' = n+1-k$. Moreover,
  \begin{equation}
    \label{eq:asymptotics-m}
    \lim_{t\to-\infty} m_{k}(t)
    = \frac{1}{\sqrt{\lambda_k}}
    = \lim_{t\to+\infty} m_{k'}(t).
  \end{equation}
  In words: asymptotically as $t\to\pm\infty$, the $k$th fastest peakon has
  velocity $1/\lambda_k$ and amplitude $1/\sqrt{\lambda_k}$.
\end{theorem}

\begin{proof}
  This is just a matter of identifying the dominant terms;
  $b_1(t) = b_1(0) \, e^{t/\lambda_1}$
  grows much faster as $t\to+\infty$ than $b_2(t)$,
  which in turn grows much faster than $b_3(t)$, etc.,
  and as $t\to-\infty$ it is the other way around.
  Thus, for example,
  $U_k \sim \Psi_{12\dots k} \, b_1 b_2 \dots b_k$
  as $t\to+\infty$. A similar analysis of $W_k$ and~$Z_k$ leads
  quickly to the stated formulas.
\end{proof}

The only difference compared to the $x_k$ asymptotics for
Degasperis--Procesi peakons
\cite[Theorem~2.25]{lundmark-szmigielski-DPlong}
is that \eqref{eq:asymptotics-x} contains an additional term
$-\frac12 \ln \lambda_k$.
Since this term cancels in the subtraction,
the phase shifts for Novikov peakons are exactly the same as for
Degasperis--Procesi peakons \cite[Theorem~2.26]{lundmark-szmigielski-DPlong}:
\begin{multline}
  \label{eq:phaseshift}
  \lim_{t\to\infty} \left( x_{k'}(t)-\frac{t}{\lambda_k} \right)
  -\lim_{t\to-\infty} \left( x_k(t)-\frac{t}{\lambda_k} \right)
  =\\
  =\sum_{i=1}^{k-1}
  \log \frac{(\lambda_i-\lambda_k)^2}{(\lambda_i+\lambda_k) \lambda_i}-
  \sum_{i=k+1}^{n}\log \frac{(\lambda_i-\lambda_k)^2}{(\lambda_i+\lambda_k) \lambda_i}.
\end{multline}

\appendix

\section{Combinatorial results}
\label{sec:combinatorics}

This appendix contains some material related to the combinatorial structure
of the constants of motion $H_1,\dots,H_n$ of the Novikov peakon ODEs;
see Section~\ref{sec:forward}, and in particular
Theorem~\ref{thm:constants-of-motion}.
Recall that
\begin{equation*}
  A(\lambda) = 1 - \lambda H_1 + \dots + (-\lambda)^n H_n
  = \det(I - \lambda TPEP),
\end{equation*}
where $I$ is the $n \times n$ identity matrix,
and $T$, $E$, $P$ are $n \times n$ matrices defined by
$T_{jk} = 1 + \sgn(j-k)$, $E_{jk} = e^{-\abs{x_j-x_k}}$,
and $P = \diag(m_1,\dots,m_n)$.
The first thing to prove is that the matrix $TPEP$ is oscillatory
if all $m_k > 0$,
which shows that the zeros of $A(\lambda)$ are positive and simple.
Then we show how to easily compute the minors of $PEP$,
and finally we prove the ``Canada Day Theorem'' (Theorem~\ref{thm:CanadaDay})
which implies that $H_k$ equals the sum of all $k \times k$ minors of $PEP$.

\subsection{Preliminaries}
\label{sec:totpos}

In this section we have collected some facts about total positivity
\cite{karlin-TP,gantmacher-krein,fomin-zelevinsky}
that will be used below.

\begin{definition}
  If $X$ is a matrix and $I$ and~$J$ are index sets,
  the submatrix $(X_{ij})_{i\in I,j\in J}$
  will be denoted by $X_{IJ}$ (or sometimes $X_{I,J}$).
  The set of $k$-element subsets of the integer interval
  $[1,n] = \{ 1,2,\dots,n \}$ will be denoted
  $\binom{[1,n]}{k}$, and elements of such a subset~$I$ will always be
  assumed to be numbered in ascending order $i_1 < \dots < i_k$.
\end{definition}

\begin{definition}
  A square matrix is said to be \emph{totally positive}
  if all its minors of all orders are positive.
  It is called \emph{totally nonnegative}
  if all its minors are nonnegative.
  A matrix is \emph{oscillatory} if it is totally nonnegative
  and some power of it is totally positive.
\end{definition}

\begin{theorem}
  \label{thm:TP-spectrum}
  All eigenvalues of a totally positive matrix are positive and of
  algebraic multiplicity one, and likewise for oscillatory matrices.
  All eigenvalues of a totally nonnegative matrix are nonnegative, but
  in general of arbitrary multiplicity.
\end{theorem}

\begin{theorem}
  \label{thm:O-TN-product}
  The product of an oscillatory matrix and a nonsingular totally
  nonnegative matrix is oscillatory.
\end{theorem}

\begin{definition}
  A \emph{planar network} $(\Gamma,\omega)$ of order~$n$ is an acyclic planar
  directed graph~$\Gamma$ with arrows going from left to right,
  with $n$ sources (vertices with outgoing arrows only) on the left side,
  and with $n$ sinks (vertices with incoming arrows only) on the right side.
  The sources and sinks are numbered $1$ to~$n$, from bottom to top, say.
  All other vertices have at least one arrow coming in and at least one arrow going out.
  Each edge~$e$ of the graph~$\Gamma$ is assigned a scalar weight~$\omega(e)$.
  The \emph{weight} of a directed path in $\Gamma$ is the product of all
  the weights of the edges of that path. The \emph{weighted path matrix}
  $\Omega(\Gamma,\omega)$ is the $n\times n$ matrix whose $(i,j)$ entry
  $\Omega_{ij}$ is the sum of the weights of the possible paths from
  source~$i$ to sink~$j$.
\end{definition}

The following theorem was discovered by
Lindstr\"om \cite{lindstrom-matroids}
and made famous by Gessel and Viennot
\cite{gessel-viennot-binomial-determinants}.
A similar theorem also appeared earlier in the work of
Karlin and McGregor on birth and death processes
\cite{karlin-mcgregor-coincidence-probabilities}.

\begin{theorem}[Lindström's Lemma]
  \label{thm:lindstrom}
  Let $I$ and $J$ be subsets of $\{1,\ldots,n\}$ with the same cardinality.
  The minor $\det\Omega_{IJ}$
  of the weighted path matrix $\Omega(\Gamma,\omega)$ of a
  planar network is equal to the sum of the weights of
  all possible families of nonintersecting paths
  (i.e., paths having no vertices in common)
  connecting the sources
  labelled by~$I$ to the sinks labelled by~$J$.
  (The weight of a family of paths is defined as
  the product of the weights of the individual paths.)
\end{theorem}

\begin{corollary}
  \label{cor:pathmatrix}
  If all weights $\omega(e)$ are nonnegative, then the weighted path
  matrix is totally nonnegative.
\end{corollary}

\begin{remark}
  Beware that having \emph{positive} weights does \emph{not} in general imply
  total positivity of the path matrix $\Omega$, since some minors
  $\det \Omega_{IJ}$ may be zero due to absence of nonintersecting
  path families from $I$ to~$J$,
  in which case $\Omega$ is only totally nonnegative.
\end{remark}

\subsection{Proof that $TPEP$ is oscillatory}
\label{sec:TPEP-oscillatory}

The matrix $T$ is the path matrix of the planar network whose
structure is illustrated below for the case $n=4$
(with all edges, and therefore all paths and families of paths, having unit weight):
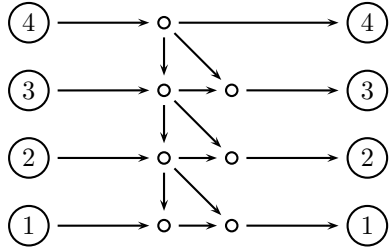
\begin{center}
  \psset{unit=9mm}
  \begin{pspicture}(-0.5,-0.5)(5.5,3.5)
    \cnodeput(0,0){source1}{$1$}
    \cnodeput(0,1){source2}{$2$}
    \cnodeput(0,2){source3}{$3$}
    \cnodeput(0,3){source4}{$4$}

    \cnodeput(5,0){sink1}{$1$}
    \cnodeput(5,1){sink2}{$2$}
    \cnodeput(5,2){sink3}{$3$}
    \cnodeput(5,3){sink4}{$4$}

    \psset{radius=0.1}
    \Cnode(2,0){in1a} \Cnode(3,0){in1b}
    \Cnode(2,1){in2a} \Cnode(3,1){in2b}
    \Cnode(2,2){in3a} \Cnode(3,2){in3b}
    \Cnode(2,3){in4a}

    \psset{nodesep=3pt}

    \ncline{->}{source4}{in4a} \ncline{->}{in4a}{sink4}
    \ncline{->}{source3}{in3a} \ncline{->}{in3a}{in3b} \ncline{->}{in3b}{sink3}
    \ncline{->}{source2}{in2a} \ncline{->}{in2a}{in2b} \ncline{->}{in2b}{sink2}
    \ncline{->}{source1}{in1a} \ncline{->}{in1a}{in1b} \ncline{->}{in1b}{sink1}
    \ncline{->}{in4a}{in3a} \ncline{->}{in4a}{in3b}
    \ncline{->}{in3a}{in2a} \ncline{->}{in3a}{in2b}
    \ncline{->}{in2a}{in1a} \ncline{->}{in2a}{in1b}
  \end{pspicture}
\end{center}
Indeed, there is clearly one path from source $i$ to sink $j$
if $i=j$, two paths if $i>j$, and none if $i<j$, and this agrees with
\begin{equation*}
  T_{ij} = 1 + \sgn(i-j) =
  \begin{cases}
    1, & i=j,\\
    2, & i>j,\\
    0, & i<j.
  \end{cases}
\end{equation*}
Similarly one can check that the matrix $PEP$ is the weighted path
matrix of the planar network illustrated below for the case $n=5$
(we are assuming that $x_1 < \dots < x_n$, so that
$E_{12} E_{23} = e^{x_1-x_2} e^{x_2-x_3} = E_{13}$, etc.):
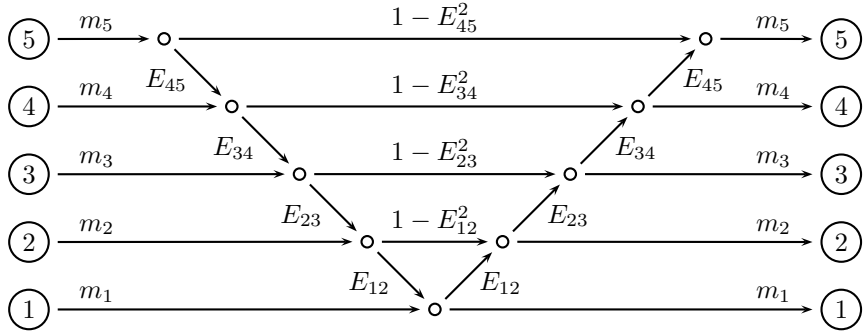
\begin{center}
  \psset{unit=9mm}
  \begin{pspicture}(-0.5,-0.5)(12.5,4.6)
    \cnodeput(0,0){source1}{$1$}
    \cnodeput(0,1){source2}{$2$}
    \cnodeput(0,2){source3}{$3$}
    \cnodeput(0,3){source4}{$4$}
    \cnodeput(0,4){source5}{$5$}

    \cnodeput(12,0){sink1}{$1$}
    \cnodeput(12,1){sink2}{$2$}
    \cnodeput(12,2){sink3}{$3$}
    \cnodeput(12,3){sink4}{$4$}
    \cnodeput(12,4){sink5}{$5$}

    \psset{radius=0.1}
    \Cnode(6,0){a}
    \Cnode(5,1){b1} \Cnode(7,1){b2}
    \Cnode(4,2){c1} \Cnode(8,2){c2}
    \Cnode(3,3){d1} \Cnode(9,3){d2}
    \Cnode(2,4){e1} \Cnode(10,4){e2}

    \psset{nodesep=3pt}
    \psset{labelsep=2pt}

    \ncline{->}{source1}{a} \taput[tpos=0.166667]{$m_1$}
    \ncline{->}{source2}{b1} \taput[tpos=0.2]{$m_2$}
    \ncline{->}{source3}{c1} \taput[tpos=0.25]{$m_3$}
    \ncline{->}{source4}{d1} \taput[tpos=0.333333]{$m_4$}
    \ncline{->}{source5}{e1} \taput[tpos=0.5]{$m_5$}

    \ncline{<-}{sink1}{a} \taput[tpos=0.8333333]{$m_1$}
    \ncline{<-}{sink2}{b2} \taput[tpos=0.8]{$m_2$}
    \ncline{<-}{sink3}{c2} \taput[tpos=0.75]{$m_3$}
    \ncline{<-}{sink4}{d2} \taput[tpos=0.666667]{$m_4$}
    \ncline{<-}{sink5}{e2} \taput[tpos=0.5]{$m_5$}

    \ncline{->}{b1}{a} \tlput[tpos=0.6]{$E_{12}$}
    \ncline{->}{a}{b2} \trput[tpos=0.6]{$E_{12}$}
    \ncline{->}{b1}{b2} \taput{$1-E_{12}^2$}

    \ncline{->}{c1}{b1} \tlput[tpos=0.6]{$E_{23}$}
    \ncline{->}{b2}{c2} \trput[tpos=0.6]{$E_{23}$}
    \ncline{->}{c1}{c2} \taput{$1-E_{23}^2$}

    \ncline{->}{d1}{c1} \tlput[tpos=0.6]{$E_{34}$}
    \ncline{->}{c2}{d2} \trput[tpos=0.6]{$E_{34}$}
    \ncline{->}{d1}{d2} \taput{$1-E_{34}^2$}

    \ncline{->}{e1}{d1} \tlput[tpos=0.6]{$E_{45}$}
    \ncline{->}{d2}{e2} \trput[tpos=0.6]{$E_{45}$}
    \ncline{->}{e1}{e2} \taput{$1-E_{45}^2$}
  \end{pspicture}
\end{center}
By Corollary~\ref{cor:pathmatrix},
both $T$ and $PEP$ are totally nonnegative (if all $m_k > 0$).
Furthermore, $(PEP)^N$ is the weighted path matrix of the planar
network obtain by connecting $N$ copies of the network for $PEP$
in series, and if $N$ is large enough, there is clearly enough
wiggle room in this network to find a nonintersecting path family
from any source set~$I$ to any sink set~$J$ with $\abs{I}=\abs{J}$.
Thus $(PEP)^N$ is totally positive for sufficiently large $N$; in
other words, $PEP$ is oscillatory.
(Another way to see this is to use a criterion
\cite[Chapter~II, Theorem~10]{gantmacher-krein}
which says that a totally nonnegative matrix $X$
is oscillatory if and only if it is nonsingular and $X_{ij} > 0$ for
$\abs{i-j}=1$.)
Since $T$ is nonsingular, Theorem~\ref{thm:O-TN-product} implies that
$TPEP$ is oscillatory, which was the first thing we wanted to prove.

\subsection{Minors of $PEP$}
\label{sec:PEP-minors}

Having a planar network for $PEP$ makes it easy to compute its
minors using Lindström's Lemma.

\begin{example}
  \label{ex:H3-when-n-equals-6}
  Consider the constant of motion $H_3$ in the case $n=6$.

  For sources $I = \{ 1,2,3 \}$ and sinks $J = \{ 1,2,3 \}$
  there is only one family of nonintersecting paths, namely the paths
  going straight across. The weights of these paths are $m_1m_1$,
  $m_2 (1-E_{12}^2) m_2$ and $m_3 (1-E_{23}^2) m_3$,
  and the total weight of that family is therefore
  $(1-E_{12}^2) (1-E_{23}^2) \, m_1^2 m_2^2 m_3^2$, which will be the first term
  in $H_3$.

  A similar term results whenever $I=J$.
  For instance, when $I = J = \{ 1,2,4 \}$ the paths starting at sources $1$
  and~$2$ must go straight across, while the path from source~$4$ to
  to sink~$4$ can go straight across, or down to line~$3$ and up again.
  The contributions from these two possible nonintersecting path families
  add up to
  \begin{multline*}
    m_1 m_1 \cdot m_2 (1-E_{12}^2) m_2 \cdot
    \Bigl( m_4 (1-E_{34}^2) m_4 + m_4 E_{34} (1-E_{23}^2) E_{34} m_4 \Bigr)
    \\[1.5ex]
    = (1-E_{12}^2) (1-E_{24}^2) \, m_1^2 m_2^2 m_4^2.
  \end{multline*}

  From $I = \{ 1,2,3 \}$ to $J = \{ 1,2,4 \}$ there is one nonintersecting
  path family, and there is another one with the same weight
  from $I = \{ 1,2,4 \}$ to $J = \{ 1,2,3 \}$; the two add up to the term
  $2 (1-E_{12}^2) (1-E_{23}^2) E_{24} \, m_1^2 m_2^2 m_3 m_4$.

  Continuing like this, one finds that the types of terms that appear
  in $H_3$ are
  \begin{equation}
    \label{eq:H3n6}
    \begin{split}
      H_3 &= (1-E_{12}^2) (1-E_{23}^2) \, m_1^2 m_2^2 m_3^2 + \ldots
      \\ & \quad +
      2 (1-E_{12}^2) (1-E_{23}^2) E_{34} \, m_1^2 m_2^2 m_3 m_4 + \ldots
      \\ & \quad +
      4 (1-E_{12}^2) (1-E_{34}^2) E_{23} E_{45} \, m_1^2 m_2 m_3 m_4 m_5 + \ldots
      \\ & \quad +
      8 \, (1-E_{23}^2) (1-E_{45}^2) E_{12} E_{34} E_{56} \, m_1 m_2 m_3 m_4 m_5 m_6.
    \end{split}
  \end{equation}
  The last term comes from the $8$ possible nonintersecting path families
  from $I = \{ i_1,i_2,i_3 \}$ to $J = \{ j_1,j_2,j_3 \}$
  where $(i_1,j_1)=(1,2)$ or $(2,1)$,
  $(i_2,j_2)=(3,4)$ or $(4,3)$,
  and $(i_3,j_3)=(5,6)$ or $(6,5)$.
\end{example}

\begin{remark}
  \label{rem:single-pair}
  Alternatively, the $m_k$ can be factored out from any minor of
  $PEP$, leaving the corresponding minor of~$E$, which can be computed
  using a result from Gantmacher and Krein
  \cite[Section~II.3.5]{gantmacher-krein},
  since the matrix $E$ is what they call a \emph{single-pair matrix}.
  This means a symmetric $n \times n$ matrix $X$ with entries
  \begin{equation}
    \label{eq:single-pair}
    X_{ij} =
    \begin{cases}
      \psi_i \chi_j, & i \le j,\\
      \psi_j \chi_i, & i \ge j.
    \end{cases}
  \end{equation}
  The $k \times k$ minors of a single-pair matrix are given by the
  following rule:
  $\det X_{IJ} = 0$, unless $I,J \in \binom{[1,n]}{k}$ satisfy the condition
  \begin{equation}
    \label{eq:single-pair-minor-condition}
    (i_1,j_1) < (i_2,j_2) < \dots < (i_k,j_k),
  \end{equation}
  where the notation means that both numbers in one pair must be less
  than both numbers in the following pair;
  in this case,
  \begin{equation}
    \label{eq:single-pair-minor}
    \det X_{IJ} =
    \psi_{\alpha_1}
    \begin{vmatrix} \chi_{\beta_1} & \chi_{\alpha_2} \\ \psi_{\beta_1} & \psi_{\alpha_2} \end{vmatrix}
    \begin{vmatrix} \chi_{\beta_2} & \chi_{\alpha_3} \\ \psi_{\beta_2} & \psi_{\alpha_3} \end{vmatrix}
    \dots
    \begin{vmatrix} \chi_{\beta_{k-1}} & \chi_{\alpha_k} \\ \psi_{\beta_{k-1}} & \psi_{\alpha_k} \end{vmatrix}
    \chi_{\beta_k},
  \end{equation}
  where
  \begin{equation}
    \label{eq:alphabeta}
    (\alpha_m,\beta_m) = \bigl( \min(i_m,j_m), \max(i_m,j_m) \bigr).
  \end{equation}
  In the case of $E$ we have $\psi_i = e^{x_i}$ and $\chi_i = e^{-x_i}$
  (assuming as usual that $x_1 < \dots < x_n$),
  and \eqref{eq:single-pair-minor} becomes
  \begin{equation}
    \label{eq:E-minor}
    \det E_{IJ} =
    (1 - E_{\beta_1 \alpha_2}^2 )
    (1 - E_{\beta_2 \alpha_3}^2 )
    \dots
    (1 - E_{\beta_{k-1} \alpha_k}^2 )
    E_{\alpha_1 \beta_1}
    E_{\alpha_2 \beta_2}
    \dots
    E_{\alpha_k \beta_k}.
  \end{equation}
\end{remark}

\subsection{Proof of the ``Canada Day Theorem''}
\label{sec:CDT-proof}

The result to be proved (Theorem~\ref{thm:CanadaDay}) is that for any
symmetric $n \times n$ matrix~$X$, the coefficient of $s^k$ in the
polynomial $\det(I+s\,TX)$ equals the sum of all $k\times k$ minors
of~$X$:
\begin{equation}
  \label{eq:CanadaDay}
  \det(I+s\,TX) =
  1 + \sum_{k=1}^n \left(
    \sum_{I \in \binom{[1,n]}{k}} \sum_{J \in \binom{[1,n]}{k}} \det X_{IJ}
  \right) s^k.
\end{equation}
We start from the elementary fact that for any matrix $Y$,
the coefficients in its characteristic polynomial
are given by the sums of the \emph{principal} minors,
\begin{equation*}
  \det(I+s\,Y) =
  1 + \sum_{k=1}^n \left(
    \sum_{J \in \binom{[1,n]}{k}} \det Y_{JJ}
  \right) s^k.
\end{equation*}
Applying this to $Y=TX$ and computing the minors of $TX$
using the Cauchy--Binet formula \cite[Ch.~I, \textsection~2]{gantmacher-matrixtheoryI}
\begin{equation}
  \label{eq:Cauchy-Binet}
  \det (TX)_{AB} = \sum_{I \in \binom{[1,n]}{k}} \det T_{AI} \, \det X_{IB},
  \qquad \text{for $A,B \in \textstyle \binom{[1,n]}{k}$},
\end{equation}
we find that
\begin{equation*}
  \det(I+s\,TX) =
  1 + \sum_{k=1}^n \left(
    \sum_{I \in \binom{[1,n]}{k}} \sum_{J \in \binom{[1,n]}{k}} \det T_{JI} \, \det  X_{IJ}
  \right) s^k.
\end{equation*}
Comparing this to \eqref{eq:CanadaDay}, it is clear that what we need
to show is that, for any~$k$,
\begin{equation}
  \label{eq:target-sum}
  \sum_{I \in \binom{[1,n]}{k}} \sum_{J \in \binom{[1,n]}{k}} \det T_{JI} \, \det  X_{IJ}
  = \sum_{I \in \binom{[1,n]}{k}} \sum_{J \in \binom{[1,n]}{k}} \det X_{IJ}.
\end{equation}
The first thing to do is calculate the minors $\det T_{JI}$.

\begin{definition}
  \label{def:interlace}
  Given $I, J \in \binom{[1,n]}{k}$, the set $I$ is said to
  \emph{interlace} with the set $J$, denoted $I\leq J$, if
  \begin{equation}
    \label{eq:IJinterlace}
    i_1 \leq j_1 \leq i_2 \leq j_2 \leq \ldots \leq i_k \leq j_k.
  \end{equation}
  If all the inequalities are strict, then $I$ is said to
  \emph{strictly interlace} with~$J$, in which case we write $I<J$.
  If $I \leq J$, then $I'$ and $J'$ 
  denote the strictly interlacing subsets (possibly empty)
  \begin{equation}
    \label{eq:IJprime}
    I' = I \setminus (I \cap J),
    \qquad
    J' = J \setminus (I \cap J),
  \end{equation}
  whose cardinality (possibly zero) will be denoted by
  \begin{equation}
    \label{eq:pIJ}
    p(I,J)= \abs{I'} = \abs{J'}.
  \end{equation}
\end{definition}

\begin{lemma}
  \label{lem:minors-T}
  For $I, J \in \binom{[1,n]}{k}$, the corresponding $k \times k$
  minor of $T$ is
  \begin{equation}
    \label{eq:minors-T}
    \det T_{JI} =
    \begin{cases}
      2^{p(I,J)}, & \text{if $I \leq J$}, \\
      0, & \text{otherwise}.
    \end{cases}
  \end{equation}
\end{lemma}

\begin{proof}
  We will use Lindström's Lemma (Theorem~\ref{thm:lindstrom})
  on the planar network for $T$ given in
  Section~\ref{sec:TPEP-oscillatory} above;
  the minor $\det T_{JI}$ equals the total number of families of
  nonintersecting paths connecting the source nodes (on the left) indexed
  by~$J$ to the sink nodes (on the right) indexed by~$I$.

  The proof proceeds by induction on the size~$n$ of~$T$. The claim is
  trivially true for $n=1$. Consider an arbitrary $n>1$, and suppose the
  claim is true for size \mbox{$n-1$}. If neither $I$ nor $J$ contain
  $n$, the claim follows immediately from the induction hypothesis,
  and likewise if $I$ and~$J$ both contain~$n$, because there is only
  one path connecting source~$n$ to sink~$n$. If $I$ contains $n$ but
  $J$ does not, then $\det T_{JI}=0$ because there are no paths going
  upward; this agrees with the claim, since in this case $I$ does not
  interlace with~$J$.

  The only remaining case is therefore
  $J = J_1 \cup \{n\}$,
  $I = I_1 \cup \{ i_k \}$, with $i_k <n$.
  But then
  \begin{equation*}
    \det T_{JI} = \det T_{J_1I_1} \times
    \begin{cases}
      2, &\text{if $j_{k-1} < i_k$}, \\
      1, &\text{if $j_{k-1} = i_k$}, \\
      0, &\text{if $j_{k-1} > i_k$},
    \end{cases}
  \end{equation*}
  depending on whether the path connecting source~$n$ with sink~$i_k$
  has to cross the $j_{k-1}$ level; if it does not, there are two
  available paths, if it does, there is only one available path
  provided $j_{k-1}=i_n$, otherwise the path intersects the path
  coming from source~$j_{k-1}$. In the last instance, $I$ does
  not interlace with~$J$, while in the other two $I\leq J$ if and only
  if $I_1\leq J_1$, thus proving the claim.
\end{proof}

According to this lemma, the structure of \eqref{eq:target-sum} (which is
what we want to prove) is
\begin{equation}
  \label{eq:target-sum-again}
  \sum_{\substack{I, J \in \binom{[1,n]}{k} \\ I \leq J}} 2^{p(I,J)} \, \det  X_{IJ}
  = \sum_{A , B \in \binom{[1,n]}{k} } \det  X_{AB},
\end{equation}
and we must show that those terms $\det X_{IJ}$ that occur more than once
on the left-hand side exactly compensate for those that are absent.
This will follow from another technical lemma:

\begin{lemma}[Relations between $k \times k$ minors of a symmetric matrix]
  \label{lem:minor-sum}
  Suppose $I, J \in \binom{[1,n]}{k}$ and $I \leq J$.
  Then, for any symmetric $n\times n$ matrix $X$,
  \begin{equation}
    \label{eq:minor-sum}
    \sum_{\substack{A,B \in \binom{I \cup J}{k} \\[0.5ex] A \cap B = I \cap J}} \det X_{AB}
    = 2^{p(I,J)} \det X_{IJ}.
  \end{equation}
\end{lemma}

Before proving Lemma~\ref{lem:minor-sum},
we will use it to finish the proof of the
main theorem. The two lemmas above show that the sum on the left-hand side of \eqref{eq:target-sum-again} equals
\begin{equation}
  \label{eq:plinka}
  \sum_{\substack{I,J \in \binom{[1,n]}{k} \\ I \leq J}}
  \!\!\! 2^{p(I,J)} \, \det  X_{IJ}
  \,\,\, = \!\!\!
  \sum_{\substack{I,J \in \binom{[1,n]}{k} \\ I \leq J}}
  \left(
    \sum_{\substack{A,B \in \binom{I \cup J}{k} \\[0.5ex] A \cap B = I \cap J}} \det X_{AB}
  \right),
\end{equation}
which in turn equals the sum on the right-hand side of
\eqref{eq:target-sum-again},
\begin{equation}
  \label{eq:plonka}
  \sum_{{A,B \in \binom{[1,n]}{k}}} \det X_{AB}.
\end{equation}
Thus \eqref{eq:target-sum-again} holds, and the theorem is proved.
The final step from \eqref{eq:plinka} to \eqref{eq:plonka} is justified
by the observation that any given pair $(A,B)$ of the type summed over in
\eqref{eq:plonka} appears
exactly once in the right-hand side of \eqref{eq:plinka},
namely for the sets $I$ and~$J$ defined as follows.
Let $M = A \cap B$, $A' = A \setminus M$, $B' = B \setminus M$,
and let $p \ge 0$ be the cardinality of the disjoint sets $A'$ and $B'$
(they are empty sets if $p=0$).
Then define $I'$ and $J'$ by enumerating the $2p$ elements of $A' \cup B'$
in the strictly interlacing order $I' < J'$, and let $I = M \cup I'$ and
$J = M \cup J'$.
Conversely, no other terms than these appear in the right hand side of
\eqref{eq:plinka}, and it is therefore indeed equal to \eqref{eq:plonka}.

\begin{proof}[Proof of Lemma~\ref{lem:minor-sum}]
  The sets $I \le J$ and $I' < J'$
  (as in Definition~\ref{def:interlace}), with
  \begin{equation*}
    \abs{I} = \abs{J} = k,
    \qquad
    \abs{I'} = \abs{J'} = p(I,J)=p,
  \end{equation*}
  will be fixed throughout the proof,
  and for convenience we also introduce
  $M = I \cap J$ and $U = I \cup J$, with $\abs{M}=k-p$
  and $\abs{U} = k+p$.
  We can assume that $p>0$, since the case $p=0$ is trivial;
  it occurs when $I=J$,
  and then both sides of \eqref{eq:minor-sum} simply equal $\det X_{II}$.

  The set $U \setminus M$ consists of the $2p$ numbers which
  belong alternatingly to~$I'$ and to~$J'$.
  The sum \eqref{eq:minor-sum} runs over all pairs of sets $(A,B)$
  obtained by splitting these $2p$ numbers into two disjoint
  $p$-sets $A'$ and $B'$ in an arbitrary way and letting
  $A = M \cup A'$ and $B = M \cup B'$.
  Write $\mathcal{Q}$ for this set;
  that is, $\mathcal{Q}$ denotes the set of pairs
  $(A,B) \in \binom{[1,n]}{k} \times \binom{[1,n]}{k}$
  such that $A \cup B = U$ and $A \cap B = M$.
  After expanding $\det X_{AB}$, we can then write the left-hand side of
  \eqref{eq:minor-sum} as
  \begin{equation}
    \label{eq:doublesum-Canada}
    \sum_{((A,B),\sigma) \in \mathcal{Q} \times \mathcal{S}_k}
    (-1)^{\sigma} X_{a_1 b_{\sigma(1)}} X_{a_2 b_{\sigma(2)}} \dots X_{a_k b_{\sigma(k)}},
  \end{equation}
  where $\mathcal{S}_k$ is the group of permutations of $\{ 1,2,\dots,k \}$,
  and $(-1)^{\sigma}$ denotes the sign of the permutation~$\sigma$.

  For each $((A,B),\sigma) \in \mathcal{Q} \times \mathcal{S}_k$,
  we let $A' = A \setminus M$ and $B' = B \setminus M$,
  and set up a ($\sigma$-dependent) bijection between $A'$ and $B'$ as follows:
  $a' \in A'$ is paired up with $b' \in B'$ if and only if   the product
  $X_{a_1 b_{\sigma(1)}} X_{a_2 b_{\sigma(2)}} \dots X_{a_k b_{\sigma(k)}}$
  contains either the factor $X_{a'b'}$ or a sequence of factors
  $X_{a'r}$, $X_{rs}$, \ldots, $X_{tb'}$ where $r,s,\dots,t \in M$.
  Let us say that $a'$ and $b'$ are \emph{linked} if they are paired
  up in this manner.
  A linked pair $(a',b') \in A' \times B'$ will be called
  \emph{hostile} if $(a',b')$ belongs to $I' \times I'$ or $J' \times J'$,
  and \emph{friendly} if $(a',b')$ belongs to $I' \times J'$ or $J' \times I'$.
  To each term in the sum \eqref{eq:doublesum-Canada}
  there will thus correspond $p$ such linked pairs,
  and what we will show is that the terms containing at least one hostile
  pair will cancel out,
  and that the remaining terms (with all friendly pairs) will add up to
  the right-hand side of \eqref{eq:minor-sum}.

  Next we define what we mean by \emph{flipping} a linked pair $(a',b')$.
  This means that, in the product
  $X_{a_1 b_{\sigma(1)}} X_{a_2 b_{\sigma(2)}} \dots X_{a_k b_{\sigma(k)}}$,
  those factors $X_{a'r} X_{rs} \dots X_{tb'}$ that link
  $a'$ to~$b'$ are replaced by $X_{b't} \dots X_{sr} X_{ra'}$,
  with all the indices in reversed order.
  (When the linking involves just a single factor $X_{a'b'}$,
  flipping means replacing it by $X_{b'a'}$.)
  Since the matrix~$X$ is symmetric, this does not change the value
  of the product, but it changes the way it is indexed.
  The number $a'$ which used to be in the first slot (in $X_{a'r}$)
  is now in the second slot (in $X_{ra'}$), and vice versa for~$b'$.
  The connecting indices $r,s,\dots,t \in M$ do not contribute to any
  change in the indexing sets, since, for example,
  the $r$ in $X_{a'r}$ is moved from the second slot to the first,
  while the other $r$ in $X_{rs}$ is moved from the first to the second.
  The new product (the result of the flipping) is therefore indexed by the sets
  \begin{equation*}
    \Bigl( A \setminus \{ a' \} \Bigr) \cup \{ b' \} =:
    \widetilde{A} = \{ \widetilde{a}_1 < \dots < \widetilde{a}_k \}
  \end{equation*}
  and
  \begin{equation*}
    \Bigl( B \setminus \{ b' \} \Bigr) \cup \{ a' \} =:
    \widetilde{B} = \{ \widetilde{b}_1 < \dots < \widetilde{b}_k \}
  \end{equation*}
  respectively, and after reordering the factors so that the first indices
  come in ascending order, it can be written
  \begin{equation*}
    X_{\widetilde{a}_1 \widetilde{b}_{\widetilde{\sigma}(1)}} X_{\widetilde{a}_2 \widetilde{b}_{\widetilde{\sigma}(2)}} \dots X_{\widetilde{a}_k \widetilde{b}_{\widetilde{\sigma}(k)}}
  \end{equation*}
  for some uniquely determined permutation $\widetilde{\sigma} \in \mathcal{S}_k$.
  Flipping a given pair thus takes $((A,B),\sigma)$
  to $((\widetilde{A},\widetilde{B}),\widetilde{\sigma})$.
  This operation is invertible, with inverse given by simply
  flipping the same pair again, now viewed as a pair
  $(b',a') \in ((\widetilde{A})',(\widetilde{B})')$
  linked via the indices $t,\dots,s,r$.
  Because of the symmetry of the matrix~$X$,
  the term in \eqref{eq:doublesum-Canada} corresponding to
  $((\widetilde{A},\widetilde{B}),\widetilde{\sigma})$
  is equal to the term corresponding to
  $((A,B),\sigma)$,
  except possibly for a difference in sign,
  depending on whether the signs of $\sigma$ and
  $\widetilde{\sigma}$ come out equal or not:
  \begin{equation*}
    (-1)^{\widetilde{\sigma}} X_{\widetilde{a}_1 \widetilde{b}_{\widetilde{\sigma}(1)}} X_{\widetilde{a}_2 \widetilde{b}_{\widetilde{\sigma}(2)}} \dots X_{\widetilde{a}_k \widetilde{b}_{\widetilde{\sigma}(k)}}
    = \pm (-1)^{\sigma} X_{a_1 b_{\sigma(1)}} X_{a_2 b_{\sigma(2)}} \dots X_{a_k b_{\sigma(k)}}.
  \end{equation*}
  We will show below that \emph{the permutation~$\widetilde{\sigma}$ has the
    same sign as~$\sigma$ when a friendly pair is flipped,
    and the opposite sign when a hostile pair is flipped}.
  Taking this for granted for the moment,
  divide the set $\mathcal{Q} \times \mathcal{S}_k$ into the two sets
  $(\mathcal{Q} \times \mathcal{S}_k)_{\mathrm{hostile}}$,
  consisting of those $((A,B),\sigma)$ for which at least one linked pair
  is hostile, and $(\mathcal{Q} \times \mathcal{S}_k)_{\mathrm{friendly}}$,
  consisting of those $((A,B),\sigma)$ for which all $p$ linked pairs
  are friendly.
  The mapping ``flip that out of all hostile pairs $(a',b')$ for which
  $\min(a',b')$ is smallest''
  is an involution on $(\mathcal{Q} \times \mathcal{S}_k)_{\mathrm{hostile}}$
  that pairs up each term with a partner term that is equal except for
  having the opposite sign (since it is a hostile pair that is flipped).
  Consequently these terms cancel out, and the contribution from
  $(\mathcal{Q} \times \mathcal{S}_k)_{\mathrm{hostile}}$ to \eqref{eq:doublesum-Canada}
  is zero. The sum therefore reduces to
  \begin{equation}
    \label{eq:doublesum-Canada-friendy}
    \sum_{((A,B),\sigma) \in (\mathcal{Q} \times \mathcal{S}_k)_{\mathrm{friendly}}}
    (-1)^{\sigma} X_{a_1 b_{\sigma(1)}} X_{a_2 b_{\sigma(2)}} \dots X_{a_k b_{\sigma(k)}}.
  \end{equation}
  Now equip the set $(\mathcal{Q} \times \mathcal{S}_k)_{\mathrm{friendly}}$
  with an equivalence relation;
  $((\widetilde{A},\widetilde{B}),\widetilde{\sigma})$
  and $((A,B),\sigma)$ are equivalent if one can go from one to another
  by flipping friendly pairs.
  Each equivalence class contains $2^p$ elements, since each of the
  $p$ friendly pairs can belong to either $I' \times J'$ or $J' \times I'$.
  Moreover, the terms corresponding to the elements in one equivalence class
  are all equal (including the sign, since only friendly pairs are flipped),
  and each class has a ``canonical'' representative with all linked pairs
  belonging to $I' \times J'$,
  \begin{equation*}
    \label{eq:canonical-term}
    (-1)^{\sigma} X_{i_1 j_{\sigma(1)}} X_{i_2 j_{\sigma(2)}} \dots X_{i_k j_{\sigma(k)}},
  \end{equation*}
  where the permutation $\sigma$ is uniquely determined by the equivalence class
  (and vice versa).
  Thus \eqref{eq:doublesum-Canada-friendy} becomes
  \begin{equation}
    \label{eq:doublesum-Canada-finished}
    2^p \sum_{\sigma \in \mathcal{S}_k}
    (-1)^{\sigma} X_{i_1 j_{\sigma(1)}} X_{i_2 j_{\sigma(2)}} \dots X_{i_k j_{\sigma(k)}}
    = 2^p \det X_{IJ},
  \end{equation}
  which is what we wanted to prove.

  To finish the proof, it now remains to demonstrate the rule that
  $\widetilde\sigma$ has the same (opposite) sign as $\sigma$
  when a friendly (hostile) pair is flipped.
  To this end, we will represent $((A,B),\sigma)$ with a bipartite
  graph, with the numbers in $U=A \cup B$ (in increasing order)
  as nodes both on the left and
  on the right, and the left nodes $a_i \in A$ connected by edges to
  the corresponding right nodes $b_{\sigma(i)} \in B$.
  The sign of $\sigma$ will then be equal to $(-1)^c$, where
  $c$ is the crossing number of the graph.
  As an aid in explaining the ideas we will use the following example
  with $U=[1,8]$,
  where the nodes in $M=A \cap B$ are marked with diamonds,
  and the nodes in $A'$ and $B'$ are marked with circles:
  
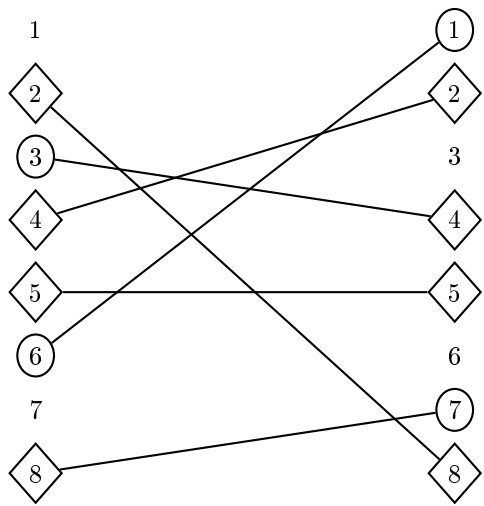
\begin{center}
    \begin{psmatrix}[colsep=20mm,rowsep=1mm]
                   1 & [mnode=oval] 1 \\[0pt]
      [mnode=dia]  2 & [mnode=dia]  2 \\[0pt]
      [mnode=oval] 3 &              3 \\[0pt]
      [mnode=dia]  4 & [mnode=dia]  4 \\[0pt]
      [mnode=dia]  5 & [mnode=dia]  5 \\[0pt]
      [mnode=oval] 6 &              6 \\[0pt]
                   7 & [mnode=oval] 7 \\[0pt]
      [mnode=dia]  8 & [mnode=dia]  8 \\[1ex]
      $\begin{aligned}
        A &= \{ 2,3,4,5,6,8 \} \\ &= \{ 2,4,5,8 \} \cup \{ 3,6 \} \\ & = M \cup A'
      \end{aligned}$
      &
      $\begin{aligned}
        B &= \{ 1,2,4,5,7,8 \} \\ &= \{ 2,4,5,8 \} \cup \{ 1,7 \} \\ & = M \cup B'
      \end{aligned}$
    \end{psmatrix}
    \ncline{2,1}{8,2}
    \ncline{3,1}{4,2}
    \ncline{4,1}{2,2}
    \ncline{5,1}{5,2}
    \ncline{6,1}{1,2}
    \ncline{8,1}{7,2}
  \end{center}
  Clearly,
  $A' \cup B' = \{ 3,6 \} \cup \{ 1,7 \} = \{ 1,3,6,7 \} = \{ i_1' < j_1' < i_2' < j_2' \}$,
  so that
  $I'=\{ i_1',i_2' \} = \{ 1,6 \}$ and $J'=\{ j_1',j_2' \} = \{ 3,7 \}$.
  Consequently, $I = M \cup I' = \{ 1,2,4,5,6,8 \}$ and
  $J = M \cup J' = \{ 2,3,4,5,6,7 \}$.
  The chosen permutation is $\sigma(123456)=632415$,
  where the notation means that $\sigma(1)=6$, $\sigma(2)=3$, etc.;
  for example,
  the latter equality comes from the second smallest number $a_2$ in~$A$
  being connected to the third smallest number $b_3$ in~$B$.
  There are $9$~crossings, so $\sigma$ is an odd permutation, and
  this graph therefore represents the term
  $-X_{28}X_{34}X_{42}X_{55}X_{61}X_{87}$,
  appearing with a minus sign in the sum \eqref{eq:doublesum-Canada}.
  The linked pairs $(a',b') \in A' \times B'$ are $(6,1)$
  (directly linked) and $(3,7)$ (linked via $4,2,8 \in M$).
  Both pairs are hostile, since $(6,1) \in I' \times I'$
  and $(3,7) \in J' \times J'$.

  We will illustrate in detail what happens when the pair
  $(3,7)$ is flipped.
  The flip is effected by replacing the factors
  $X_{34}X_{42}X_{28}X_{87}$ by $X_{78}X_{82}X_{24}X_{43}$
  and sorting the resulting product so that the first indices come
  in ascending order; this gives
  $X_{24}X_{43}X_{55}X_{61}X_{78}X_{82}$.
  Thus
  $\widetilde{A} = \{ 2,4,5,6,7,8 \}$,
  $\widetilde{B} = \{ 1,2,3,4,5,8 \}$,
  and
  $\widetilde{\sigma}(123456)=435162$
  (an even permutation).
  In terms of the graph, the nodes that are involved in the flip are,
  on both sides,
  $\{ 2,3,4,7,8 \}$ (the two nodes in the pair being flipped,
  plus the nodes linking them),
  and the edges involved are $\{ 34, 42, 28, 87 \}$, which get changed into
  $\{ 43, 24, 82, 78 \}$. In other words, the flip corresponds
  to this \emph{active subgraph} being mirror reflected across the
  central vertical line.
  To understand how the process of reflection affects the crossing
  number, it can be broken down into two steps, as follows.

  On the left, node~$7$ is unoccupied to begin with,
  so we can change the edge $87$ to $77$.
  This frees node~$8$ on the left, so that we can change the edge $28$ to $88$,
  which frees node~$2$ on the left.
  (Think of this edge as a rubber band connected at one end
  to node $8$ on the right;
  we're disconnecting its other end from node $2$ on the left and
  sliding it past all the other nodes down to node $8$ on the left.
  Obviously the crossing number increases or decreases by one every time
  we slide past a node that has an edge attached to it.)
  Continuing like this, we get the result illustrated in Step~1 below;
  the edges changed are
  $87\to77$,
  $28\to88$,
  $42\to22$,
  $34\to44$.
  \begin{center}
    \mbox{}\hfill
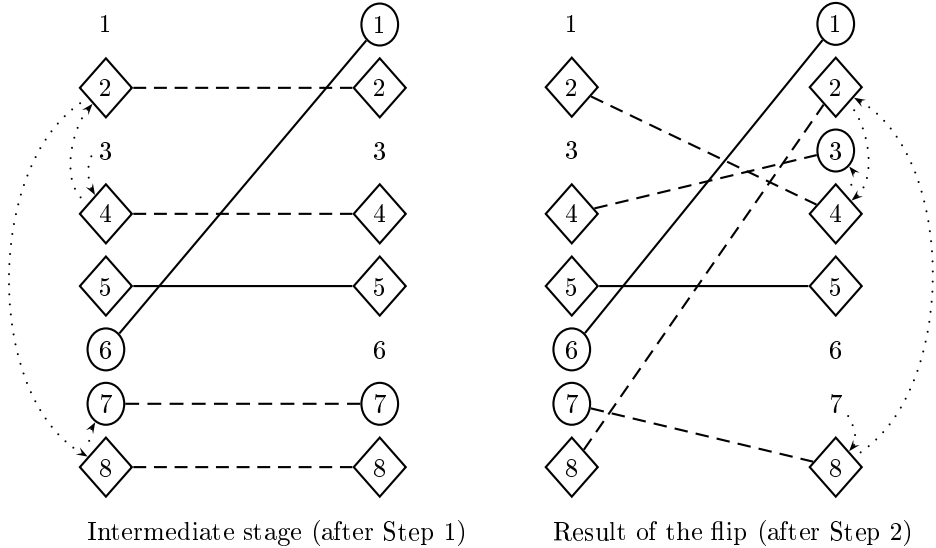
    \begin{psmatrix}[colsep=20mm,rowsep=1mm]
                   1 & [mnode=oval] 1 \\[0pt]
      [mnode=dia]  2 & [mnode=dia]  2 \\[5pt]
                   3 &              3 \\[5pt]
      [mnode=dia]  4 & [mnode=dia]  4 \\[0pt]
      [mnode=dia]  5 & [mnode=dia]  5 \\[0pt]
      [mnode=oval] 6 &              6 \\[0pt]
      [mnode=oval] 7 & [mnode=oval] 7 \\[0pt]
      [mnode=dia]  8 & [mnode=dia]  8 \\[1ex]
      \psspan{2} \hfill Intermediate stage (after Step 1) \hfill
    \end{psmatrix}
    \ncline[linestyle=dashed]{2,1}{2,2}
    \ncline[linestyle=dashed]{8,1}{8,2}
    \ncline[linestyle=dashed]{4,1}{4,2}
    \ncline[linestyle=dashed]{7,1}{7,2}
    \ncline{5,1}{5,2}
    \ncline{6,1}{1,2}
    \nccurve[linestyle=dotted,angleA=-180,angleB=120]{->}{3,1}{4,1}
    \nccurve[linestyle=dotted,angleA=150,angleB=-130]{->}{4,1}{2,1}
    \nccurve[linestyle=dotted,angleA=-150,angleB=150]{->}{2,1}{8,1}
    \nccurve[linestyle=dotted,angleA=130,angleB=-120]{->}{8,1}{7,1}
    \hfill
    \begin{psmatrix}[colsep=20mm,rowsep=1mm]
                   1 & [mnode=oval] 1 \\[0pt]
      [mnode=dia]  2 & [mnode=dia]  2 \\[0pt]
                   3 & [mnode=oval] 3 \\[0pt]
      [mnode=dia]  4 & [mnode=dia]  4 \\[0pt]
      [mnode=dia]  5 & [mnode=dia]  5 \\[0pt]
      [mnode=oval] 6 &              6 \\[0pt]
      [mnode=oval] 7 &              7 \\[0pt]
      [mnode=dia]  8 & [mnode=dia]  8 \\[1ex]
      \psspan{2} \hfill Result of the flip (after Step 2) \hfill
    \end{psmatrix}
    \ncline[linestyle=dashed]{2,1}{4,2}
    \ncline[linestyle=dashed]{8,1}{2,2}
    \ncline[linestyle=dashed]{4,1}{3,2}
    \ncline[linestyle=dashed]{7,1}{8,2}
    \ncline{5,1}{5,2}
    \ncline{6,1}{1,2}
    \nccurve[linestyle=dotted,angleA=-50,angleB=60]{<-}{3,2}{4,2}
    \nccurve[linestyle=dotted,angleA=40,angleB=-50]{<-}{4,2}{2,2}
    \nccurve[linestyle=dotted,angleA=-30,angleB=30]{<-}{2,2}{8,2}
    \nccurve[linestyle=dotted,angleA=50,angleB=-40]{<-}{8,2}{7,2}
  \end{center}
  In Step~2, we work similarly on the right-hand side:
  node~$3$ is unoccupied to begin with, so we can change edge $44$ to $43$,
  and so on.
  The list of edge moves is $44\to43$, $22\to24$, $88\to82$, $77\to78$.
  (In the graph on the right we see that the crossing number after the flip is~$8$,
  verifying the claim that $\widetilde{\sigma}$ is an even permutation.)

  We need to keep track of the changes in the crossing number
  caused by sliding active edges past nodes that have edges attached to them.
  This is most easily done by following the dotted lines
  in the figures, and counting whether the nodes that are marked
  (with circles and diamonds) are passed an even or an odd number of times.
  However, since the active subgraph simply gets reflected, the crossings
  among its edges will be the same before and after the flip,
  so we need in fact only count how many times we pass a \emph{passive}
  marked node.
  (The passive nodes in the example are $\{ 1,5,6 \}$.)

  If a passive node belonging to $M$ is passed in Step~1, then it is
  passed the same number of times in Step~2 as well, since the nodes
  in $M$ are marked both on the left and on the right.
  Therefore they do not affect the parity of the crossing number either,
  and we can ignore the nodes marked with diamonds, and only look at
  the \emph{passive circled} nodes (all the nodes in $A'$ and $B'$
  except for the two active nodes that are being flipped).

  Passive nodes belonging to $A'$ are counted only in Step~1 and
  passive nodes in~$B'$ only in Step~2; they get counted an \emph{odd}
  number of times if they lie \emph{between} the two flipped nodes
  (like node~$6$ in the example, counted once), and an \emph{even}
  number of times otherwise (like node~$1$, never counted).
  Consequently, what determines whether the parity of the crossing number
  changes is the number of nodes \emph{between} the flipped ones that
  belong to $A' \cup B' = I' \cup J'$.
  And for a friendly pair, this number is even, while for a hostile pair,
  it is odd.

  This shows that the crossing number keeps its parity (so that
  $(-1)^{\sigma} = (-1)^{\widetilde\sigma}$) when a friendly pair
  is flipped, and the opposite when a hostile pair is flipped.
  The proof is finally complete.
\end{proof}

\section{Verification of the Lax pair for peakons}
\label{sec:multiplication}

The purpose of this appendix is to carefully verify that the Lax pair
formulation \eqref{eq:lax-x}--\eqref{eq:lax-t} of the Novikov equation
really is valid for the class of distributional solutions that we are
considering. This is not at all obvious, as should be clear from the
computations below.

\subsection{Preliminaries}
\label{sec:weak-notation}

We will need to be more precise regarding the notation here than in
the main text.
A word of warning right away:
our notation for derivatives here
will differ from that used in the rest of the paper
(where subscripts should be interpreted as distributional derivatives).

To begin with,
given $n$ smooth functions $x=x_k(t)$ such that $x_1(t) < \dots < x_n(t)$,
let $x_0(t) = -\infty$ and $x_{n+1}(t) = +\infty$,
and let $\Omega_k$ (for $k=0,\dots,n$)
denote the region $x_{k}(t) < x < x_{k+1}(t)$ in the $(x,t)$ plane.

Our computations will deal with a class that we denote $\piecewiseclass$,
consisting of piecewise smooth functions
$f(x,t)$ such that the restriction of $f$ to each region $\Omega_k$ is
(the restriction to $\Omega_k$ of) a smooth
function $f^{(k)}(x,t)$ defined on an open neighbourhood
of~$\overline{\Omega}_k$ (so that $f^{(k)}$ and its partial derivatives make
sense on the curves $x=x_k(t)$).
For each fixed $t$, the function $f(\cdot,t)$ defines a regular distribution
$T_f$ in the class~$\spaceDprime$, depending parametrically on $t$
(and written $T_f(t)$ where needed).
After having made clear exactly what is meant, we will mostly be less strict,
and write $f$ instead of $T_f$ for simplicity.

The values of $f$ on the curves $x=x_k(t)$ need not be defined; the
function defines the same distribution $T_f$ no matter what values are
assigned to $f(x_k(t),t)$.
But our assumptions imply that the left and right limits of $f$ exist,
and (suppressing the time dependence) they will be denoted by
$f(x_k^-) := f^{(k-1)}(x_k)$ and $f(x_k^+) := f^{(k)}(x_k)$,
respectively.
The jump and the average of $f$ at~$x_k$ will be denoted by
\begin{equation}
  \label{eq:jump-avg-notation}
  \jump{f(x_k)} := f(x_k^+) - f(x_k^-)
  \qquad\text{and}\qquad
  \avg{f(x_k)} := \frac{f(x_k^+) + f(x_k^-)}{2},
\end{equation}
respectively. They satisfy the product rules
\begin{equation}
  \label{eq:jump-avg-products}
  \jump{fg} = \avg{f} \jump{g} + \jump{f} \avg{g},
  \qquad
  \avg{fg} = \avg{f} \avg{g} + \smallfrac14 \jump{f} \jump{g}.
\end{equation}

We will use subscripts to denote partial derivatives in the classical
sense, so that (for example) $f_x$ denotes the piecewise smooth function
whose restriction to $\Omega_k$ is given by $\partial f^{(k)} / \partial x$
(and whose values at $x=x_k(t)$ are in general undefined).
On the other hand, $D_x$ will denote the distributional derivative,
which in addition picks up Dirac delta contributions from
jump discontinuities of~$f$ at the curves $x=x_k(t)$.
That is,
$D_x T_f = T_{f_x} + \sum_{k=1}^n \jump{f(x_k)} \delta_{x_k}$,
or, in less strict notation,
\begin{equation}
  \label{eq:Dx-jump}
  D_x f = f_x + \sum_{k=1}^n \jump{f(x_k)} \delta_{x_k}.
\end{equation}
The time derivative $D_t$ is defined as a limit in $\spaceDprime$,
\begin{equation}
  \label{eq:def-Dt}
  D_t T_f(t) = \lim_{h \to 0} \frac{T_f(t+h) - T_f(t)}{h},
\end{equation}
and it commutes with~$D_x$ by the continuity of~$D_x$ on~$\spaceDprime$.
For our class $\piecewiseclass$ of piecewise smooth functions, we have
$D_t T_f = T_{f_t} - \sum_{k=1}^n \dot x_k \jump{f(x_k)} \delta_{x_k}$,
or simply
\begin{equation}
  \label{eq:Dt-jump}
  D_t f = f_t - \sum_{k=1}^n \dot x_k \jump{f(x_k)} \delta_{x_k},
\end{equation}
where $\dot x_k = dx_k / dt$.
We also note that
$\frac{d}{dt} f(x_k^{\pm}(t),t) = f_x(x_k^{\pm}(t),t) \, \dot x_k(t) + f_t(x_k^{\pm}(t),t)$,
which gives
\begin{equation}
  \label{eq:chainrule-jump-avg}
  \begin{split}
    \smallfrac{d}{dt} \jump{f(x_k)} &= \jump{f_x(x_k)} \, \dot x_k + \jump{f_t(x_k)},
    \\
    \smallfrac{d}{dt} \avg{f(x_k)} &= \avg{f_x(x_k)} \, \dot x_k + \avg{f_t(x_k)}.
  \end{split}
\end{equation}

\subsection{The problem of multiplication}

If the function $f$ is continuous at $x=x_k$, then the Dirac delta at~$x_k$
can be multiplied by the corresponding
distribution~$T_f$
according to the well-known formula
\begin{equation}
  \label{eq:delta-product-classical}
  T_f \, \delta_x = f(x_k) \, \delta_{x_k}.
\end{equation}
But below we will have to consider this product for functions in
the class~$\piecewiseclass$
described above, where the value $f(x_k)$ is not defined.
It will turn out that in the present context, the right thing to do
is to use the average value of $f$ at the jump, and thus define
$T_f \, \delta_x := \avg{f(x_k)} \, \delta_{x_k}$.
However, since we want this to be a consequence of the analysis,
rather than an a priori assumption, we will, to begin with, just
assign a hypothetical value $f(x_k)$ and use that value in
\eqref{eq:delta-product-classical}. This assignment is justified in
the present context, as we will see below. However, we are not
claiming that this addresses any of the deeper issues; for example,
this assignment does not respect the product structure of piecewise
continuous functions.
See \cite[Ch.~5]{schwartz-distributions}
for more information about the structural
problems associated with any attempt to define a product of
distributions in $\spaceDprime$.

\subsection{Distributional Lax pair}

Peakon solutions
\begin{equation}
  \label{eq:weak-u}
  u(x,t) = \sum_{k=1}^n m_k(t) \, e^{-\abs{x-x_k(t)}}
\end{equation}
belong to the piecewise smooth class~$\piecewiseclass$.
They are continuous and satisfy
\begin{equation*}
  \begin{split}
    D_x u &= u_x = \sum_{k=1}^n m_k \, \sgn(x_k-x) \, e^{-\abs{x-x_k}},\\
    D_x^2 u &= D_x(u_x) = u_{xx} + \sum_{k=1}^n \jump{u_x(x_k)} \, \delta_{x_k}
    = u + \sum_{k=1}^n (-2m_k) \, \delta_{x_k},
  \end{split}
\end{equation*}
which implies
\begin{equation}
  \label{eq:weak-m}
  m := u - D_x^2 u = 2 \sum_{k=1}^n m_k \, \delta_{x_k}.
\end{equation}
The Lax pair \eqref{eq:lax-x}--\eqref{eq:lax-t} will
involve the functions $u$ and $D_x u$, as well as the purely
singular distribution~$m$. We will take
$\psi_1$, $\psi_2$, $\psi_3$ to be functions in $\piecewiseclass$,
and separate the regular (function) part from the singular
(Dirac delta) part. The formulation obtained in this way reads
\begin{equation}
  \label{eq:weak-lax}
  D_x \Psi = \widehat{L} \Psi,
  \qquad
  D_t \Psi = \widehat{A} \Psi,
\end{equation}
where $\Psi = (\psi_1,\psi_2,\psi_3)^t$,
\begin{equation}
  \label{eq:weak-lax-x}
  \widehat{L} = L + 2z \left( \sum_{k=1}^n m_k \, \delta_{x_k} \right) N,
  \quad
  L =
  \begin{pmatrix}
    0 & 0 & 1 \\
    0 & 0 & 0 \\
    1 & 0 & 0
  \end{pmatrix},
  \quad
  N =
  \begin{pmatrix}
    0 & 1 & 0 \\
    0 & 0 & 1 \\
    0 & 0 & 0
  \end{pmatrix},
\end{equation}
and
\begin{equation}
  \label{eq:weak-lax-t}
  \widehat{A} = A - 2z \left( \sum_{k=1}^n m_k \, u(x_k)^2 \delta_{x_k} \right) N,
  \quad
  A =
  \begin{pmatrix}
    -u u_x & u_x/z & u_x^2 \\
    u /z & -1/z^2 & -u_x/z \\
    -u^2 & u/z & u u_x
  \end{pmatrix}.
\end{equation}
Note that \eqref{eq:weak-lax} involves multiplying
$N \Psi = (\psi_2,\psi_3,0)$ by $\delta_{x_k}$,
and some value $\psi_2(x_k)$ must be assigned
in order for this to be well-defined
(we will soon see that $\psi_3$ must be continuous and therefore it is
only $\psi_2$ that presents any problems).

\begin{theorem}
  \label{thm:distributionalLaxpair}
  Provided that the product $m \psi_2$ is defined using the
  average value $\psi_2(x_k) := \avg{\psi_2(x_k)}$ at the jumps,
  \begin{equation}
    \label{eq:psi-two-use-average}
    m \psi_2 := 2 \sum_{k=1}^n m_k \, \avg{\psi_2(x_k)} \, \delta_{x_k},
  \end{equation}
  the following statement holds.
  With $u$ and $m$ given by \eqref{eq:weak-u}--\eqref{eq:weak-m},
  and with $\Psi \in \piecewiseclass$,
  the Lax pair \eqref{eq:weak-lax}--\eqref{eq:weak-lax-t}
  satisfies the compatibility condition $D_t D_x \Psi = D_x D_t \Psi$
  if and only if the peakon ODEs \eqref{eq:novikov-ode} are satisfied:
  $\dot x_k = u(x_k)^2$ and
  $\dot m_k = -m_k \, u(x_k) \avg{u_x(x_k)}$.
\end{theorem}

\begin{proof}
  For simplicity, we will write just $\sum$ instead of $\sum_{k=1}^n$.
  Identifying coefficients of $\delta_{x_k}$ in the two Lax equations
  \eqref{eq:weak-lax} immediately gives $\jump{\Psi(x_k)} = 2z \,m_k N \Psi(x_k)$
  and $-\dot x_k \jump{\Psi(x_k)} = -2z \,m_k u(x_k)^2 N \Psi(x_k)$, respectively.
  Thus, $[\psi_3(x_k)]=0$ (in other words,
  $\psi_3$ is continuous) and $\dot x_k = u(x_k)^2$.
  Next we compute the derivatives of \eqref{eq:weak-lax}:
  \begin{equation*}
    \begin{split}
      D_t(D_x \Psi) &=
      D_t(L \Psi + 2z \left( \sum m_k \, \delta_{x_k} \right) N \Psi) \\
      &= L (\widehat{A} \Psi) + 2zN \sum \smallfrac{d}{dt} \bigl( m_k \Psi(x_k) \bigr) \, \delta_{x_k} - 2zN \sum m_k \Psi(x_k) \dot x_k \delta'_{x_k},
      \\
      D_x(D_t \Psi) &=
      D_x(A\Psi - 2z \left( \sum m_k \, u(x_k)^2 \delta_{x_k} \right) N \Psi) \\
      &= (A\Psi)_x + \sum \jump{A\Psi(x_k)} \delta_{x_k} - 2zN \sum m_k \Psi(x_k) u(x_k)^2 \delta'_{x_k}.
    \end{split}
  \end{equation*}
  The regular part of \eqref{eq:weak-lax} gives $\Psi_x = L\Psi$,
  so that $(A\Psi)_x=A_x \Psi + AL\Psi$, and it is easily verified
  that $LA=A_x+AL$ holds identically (since $u_{xx}=u$).
  This implies that the regular parts of the two expressions above are equal,
  and the terms involving $\delta'_{x_k}$ are also equal since
  $\dot x_k = u(x_k)^2$.
  Therefore the compatibility condition
  $D_t(D_x \Psi) = D_x(D_t \Psi)$
  reduces to an equality between the coefficients of~$\delta_{x_k}$,
  \begin{equation}
    \label{eq:gazonk}
    -2z \, m_k u(x_k)^2 LN \Psi(x_k) + 2zN \smallfrac{d}{dt} \bigl( m_k \Psi(x_k) \bigr)
    = \jump{A\Psi(x_k)}.
  \end{equation}
  Using the product rule \eqref{eq:jump-avg-products},
  the expression for $\jump{\Psi(x_k)}$ above, and $\jump{u_x(x_k)} = -2m_k$,
  we find that the right-hand side of \eqref{eq:gazonk} equals
  \begin{multline}
    \avg{A(x_k)} \, 2z\,m_k N \Psi(x_k) + \jump{A(x_k)} \avg{\Psi(x_k)}
    =
    \\
    2z\,m_k \left(
      \begin{smallmatrix}
        0 & -u \avg{u_x} & \avg{u_x}/z \\
        0 & u/z & -1/z^2 \\
        0 & -u^2 & u/z
      \end{smallmatrix}
    \right)_{\!\! x_k}
    \!\!\! \Psi(x_k)
    +
    2 m_k \left(
      \begin{smallmatrix}
        u & -1/z & -2\avg{u_x} \\
        0 & 0 & 1/z \\
        0 & 0 & -u
      \end{smallmatrix}
    \right)_{\!\! x_k}
    \!\!\! \avg{\Psi(x_k)}.
  \end{multline}
  The (3,2) entry $-u^2$ in the matrix in the first term
  will cancel against the whole first term on the left-hand side of
  \eqref{eq:gazonk}, since the only nonzero entry of $LN$ is $(LN)_{32}=1$.
  Thus \eqref{eq:gazonk} is equivalent to
  \begin{multline}
    \label{eq:gazonk2}
    \dot m_k \, N \Psi(x_k) + m_k \, N \smallfrac{d}{dt} \Psi(x_k) =
    \\
    m_k \left(
      \begin{smallmatrix}
        0 & -u \avg{u_x} & \avg{u_x}/z \\
        0 & u/z & -1/z^2 \\
        0 & 0 & u/z
      \end{smallmatrix}
    \right)_{\!\! x_k}
    \!\!\! \Psi(x_k)
    + m_k \left(
      \begin{smallmatrix}
        u/z & -1/z^2 & -2\avg{u_x}/z \\
        0 & 0 & 1/z^2 \\
        0 & 0 &- u/z
      \end{smallmatrix}
    \right)_{\!\! x_k}
    \!\!\! \avg{\Psi(x_k)}.
  \end{multline}
  To make it clear how the assumption \eqref{eq:psi-two-use-average}
  enters the proof, we want to avoid assigning a value to $\psi_2(x_k)$
  for as long as possible.
  Therefore we can't compute $\frac{d}{dt}\Psi(x_k)$ quite yet.
  But $\avg{\Psi(x_k)}$ is well-defined, and its time derivative can
  be computed using $\Psi_x=L\Psi$ and $\Psi_t=A\Psi$ in
 \eqref{eq:chainrule-jump-avg}:
  \begin{equation*}
    \begin{split}
      N \smallfrac{d}{dt} \avg{\Psi(x_k)} &=
      N \avg{L\Psi(x_k)} \, \dot x_k + N \avg{A\Psi(x_k)}
      \\
      &= N \Bigl( L u(x_k)^2 + \avg{A(x_k)} \Bigr) \avg{\Psi(x_k)} + N \smallfrac14 \jump{A(x_k)} \jump{\Psi(x_k)}
      \\
      &=
      \left(
        \begin{smallmatrix}
          u/z & -1/z^2 & -\avg{u_x}/z \\ 0 & u/z & u \avg{u_x} \\ 0 & 0 & 0
        \end{smallmatrix}
      \right)_{\!\! x_k}
      \avg{\Psi(x_k)}
      + \smallfrac14 \underbrace{N \jump{A(x_k)}N}_{=0} 2z\,m_k \Psi(x_k).
    \end{split}
  \end{equation*}
  A bit of manipulation using this result, as well as
  $\avg{\psi_3}(x_k) = \psi_3(x_k)$, shows that the
  compatibility condition \eqref{eq:gazonk2} can be written as
  \begin{multline}
    \label{eq:gazonk3}
    m_k N \smallfrac{d}{dt} \Bigl( \Psi(x_k) - \avg{\Psi(x_k)} \Bigr)
    + \Bigl( \dot m_k + m_k u(x_k) \avg{u_x(x_k)} \Bigr) N \Psi(x)
    \\
    = m_k
    \left(
      \begin{smallmatrix}
        0 & 0 & 0\\ 0 & u/z &0 \\ 0 & 0 & 0
      \end{smallmatrix}
    \right)_{\!\! x_k}
    \Bigl( \Psi(x_k) - \avg{\Psi(x_k)} \Bigr)
  \end{multline}
  The third row is zero, and the first two rows say that
  \begin{equation*}
    \begin{split}
      \Bigl( \dot m_k + m_k u(x_k) \avg{u_x(x_k)} \Bigr) \psi_2(x_k)
      &= -m_k \smallfrac{d}{dt} \Bigl( \psi_2(x_k) - \avg{\psi_2(x_k)} \Bigr),
      \\
      \Bigl( \dot m_k + m_k u(x_k) \avg{u_x(x_k)} \Bigr) \psi_3(x_k)
      &= \smallfrac{1}{z} m_k \, u(x_k) \Bigl( \psi_2(x_k) - \avg{\psi_2(x_k)} \Bigr).
    \end{split}
  \end{equation*}
  At this point we choose to assign $\psi_2(x_k) := \avg{\psi_2}(x_k)$,
  and then it is clear that \eqref{eq:gazonk3} is satisfied if and only if
  \begin{equation*}
    \dot m_k = - m_k u(x_k) \avg{u_x(x_k)}.
  \end{equation*}
\end{proof}

\section*{Acknowledgements}

HL is supported by the Swedish Research Council (Veten\-skaps\-r{\aa}det),
and JS by the National Sciences and Engineering Research Council of Canada (NSERC).

HL and JS would like to acknowledge the hospitality of
the Mathematical Research and Conference Center, B\k{e}dlewo, Poland,
%Matematyczny O\'srodek Konferencyjny
and the Department of Mathematics and Statistics, University of Saskatchewan.

\bibliographystyle{plain}
\bibliography{novikov_peakons}

\end{document}